\documentclass[letterpaper,superscriptaddress,twocolumn]{revtex4-1}

\usepackage{color}
\usepackage{amssymb,bm}
\usepackage{amsthm}
\usepackage{amsmath}
\usepackage{dsfont}
\usepackage{algorithm}

\usepackage{threeparttable}
\usepackage{scrextend}
\usepackage{algorithmic}
\usepackage{nicefrac}
\usepackage{xcolor}
\usepackage{multirow}
\usepackage{booktabs}
\usepackage{makecell}
\usepackage{adjustbox}
\usepackage{graphicx}

\usepackage[T1]{fontenc}
\usepackage[utf8]{inputenc}
\usepackage{graphicx}
\usepackage{xr}
\usepackage[colorlinks]{hyperref}
\hypersetup{colorlinks=true, urlcolor=blue}

\usepackage{color}
\usepackage{bbm}

\newcommand{\cA}{\mathcal{A}}
\newcommand{\cB}{\mathcal{B}}
\newcommand{\cC}{\mathcal{C}}

\newcommand{\cE}{\mathcal{E}}

\newcommand{\cH}{\mathcal{H}}

\newcommand{\cL}{\mathcal{L}}

\newcommand{\cN}{\mathcal{N}}

\newcommand{\cR}{\mathcal{R}}
\newcommand{\cS}{\mathcal{S}}
\newcommand{\cT}{\mathcal{T}}

\newcommand{\bC}{\mathbb{C}}

\newcommand{\bN}{\mathbb{N}}

\newcommand{\Id}{\mathds{1}}

\newcommand{\y}{\mathbf{y}}

\newcommand{\tens}{\otimes}

\newcommand{\Tr}[1]{\mathrm{Tr}\left[#1\right]}

\newcommand{\ket}[1]{\lvert#1\rangle}
\newcommand{\bra}[1]{\langle#1\rvert}
\newcommand{\braket}[2]{\langle#1\rvert#2\rangle}
\newcommand{\ketbra}[2]{\lvert#1\rangle\langle#2\rvert}

\newcommand{\abs}[1]{\left\lvert#1\right\vert}
\newcommand{\norm}[2]{\left\|#1\right\|_{#2}}
\newcommand{\ws}{\hspace{0.5em}}
\newcommand{\ci}{\mathrm{i}}

\def\>{\rangle} 
\def\<{\langle}  

\makeatletter
\def\norm{\@ifnextchar[{\@with}{\@without}}
\def\@with[#1]#2{\left\|#2\right\|_{#1}}
\def\@without#1{\left\|#1\right\|}
\makeatother

\newtheorem{thm}{Theorem}

\newtheorem{cor}{Corollary}
\newtheorem{lem}{Lemma}

\newtheorem{innercustomlem}{Lemma}
\newenvironment{customlem}[1]
{\renewcommand\theinnercustomlem{#1}\innercustomlem}
{\endinnercustomlem}

\newtheorem{innercustomcor}{Corollary}
\newenvironment{customcor}[1]
{\renewcommand\theinnercustomcor{#1}\innercustomcor}
{\endinnercustomcor}

\AtBeginDocument{
\heavyrulewidth=.08em
\lightrulewidth=.05em
\cmidrulewidth=.03em
\belowrulesep=.65ex
\belowbottomsep=0pt
\aboverulesep=.4ex
\abovetopsep=0pt
\cmidrulesep=\doublerulesep
\cmidrulekern=.5em
\defaultaddspace=.5em
}

\usepackage{soul,xcolor}

\begin{document}
\title{Optimal Provable Robustness of Quantum Classification via Quantum Hypothesis Testing}

\author{Maurice Weber} 
\affiliation{Department of Computer Science, ETH Zürich, Universitätstrasse 6, 8092 Zürich, Switzerland} 

\author{Nana Liu}
\affiliation{Institute of Natural Sciences, Shanghai Jiao Tong University, Shanghai 200240, China}
\affiliation{Ministry of Education, Key Laboratory in Scientific and Engineering Computing, Shanghai Jiao Tong University, Shanghai 200240, China}
\affiliation{University of Michigan-Shanghai Jiao Tong University Joint Institute, Shanghai 200240, China}

\author{Bo Li}
\affiliation{Department of Computer Science, University of Illinois, Urbana, Illinois 61801, USA}

\author{Ce Zhang}
\email{ce.zhang@inf.ethz.ch}
\affiliation{Department of Computer Science, ETH Zürich, Universitätstrasse 6, 8092 Zürich, Switzerland}

\author{Zhikuan Zhao}
\email{zhikuan.zhao@inf.ethz.ch}
\affiliation{Department of Computer Science, ETH Zürich, Universitätstrasse 6, 8092 Zürich, Switzerland}

\begin{abstract}

Quantum machine learning models have the potential to offer speedups and better predictive accuracy compared to their classical counterparts. 
However, these quantum algorithms, like their classical counterparts, have been shown to also be vulnerable to input perturbations, in particular for classification problems. 
These can arise either from noisy implementations or, as a worst-case type of noise, adversarial attacks. 
In order to develop defence mechanisms and to better understand the reliability of these algorithms, it is crucial to understand their robustness properties in presence of natural noise sources or adversarial manipulation. 
From the observation that measurements involved in quantum classification algorithms are naturally probabilistic, we uncover and formalize a fundamental link between binary quantum hypothesis testing and provably robust quantum classification. 
This link leads to a tight robustness condition which puts constraints on the amount of noise a classifier can tolerate, independent of whether the noise source is natural or adversarial. Based on this result, we develop practical protocols to optimally certify robustness.
Finally, since this is a robustness condition against worst-case types of noise, our result naturally extends to scenarios where the noise source is known. Thus, we also provide a framework to study the reliability of quantum classification protocols beyond the adversarial, worst-case noise scenarios.
\end{abstract}

\maketitle

\section{Introduction}
    The flourishing interplay between quantum computation and machine learning has inspired a wealth of algorithmic invention in recent years \cite{dunjko2018machine, biamonte2017quantum, schuld2015}. Among the most promising proposals are quantum classification algorithms which aspire to leverage the exponentially large Hilbert space uniquely accessible to quantum algorithms to either drastically speed up computational bottlenecks in classical protocols \cite{zhao2019bayesian, cong2019quantum, rebentrost2014quantum, farhi2018classification}, or to construct quantum-enhanced kernels that are practically prohibitive to compute classically~\cite{havlicek2019, schuld2019quantum, lloyd2020quantum}.
    Although these quantum classifiers are recognised as having the potential to offer quantum speedup or superior predictive accuracy, they are shown to be just as vulnerable to input perturbations as their classical counter-parts \cite{lu2019,liu2020, szegedy2014,goodfellow2015}. These perturbations can occur either due to imperfect implementation which is prevalent in the noisy, intermediate-scale quantum (NISQ) era~\cite{preskill2018quantum}, or, more menacingly, due to adversarial attacks where a malicious party aims to fool a classifier by carefully crafting practically undetectable noise patterns which trick a model into misclassifying a given input.
   
    In order to address these short-comings in reliability and security of quantum machine learning, several protocols in the setting of adversarial quantum learning, i.e. learning under the worst-case noise scenario, have been developed~\cite{liu2020,lu2019, wiebe2018,du2020quantum,guan2020robustness}. More recently, data encoding schemes are linked to robustness properties of classifiers with respect to different noise models in Ref.~\cite{larose2020robust}. The connection between provable robustness and quantum differential privacy is investigated in Ref.~\cite{du2020quantum}, where naturally occurring noise in quantum systems is leveraged to increase robustness against adversaries. A further step towards robustness guarantees is made in Ref.~\cite{guan2020robustness} where a bound is derived from elementary properties of the trace distance. These advances, though having accumulated considerable momentum toward a coherent strategy for protecting quantum machine learning algorithms against adversarial input perturbations, have not yet provided an adequate framework for deriving a tight robustness condition for any given quantum classifier. In other words, the known robustness conditions are sufficient but not, in general, necessary. 
   
    Thus, a major open problem remains which is significant on both the conceptual and practical levels. Conceptually, adversarial robustness, being an intrinsic property of the classification algorithms under consideration, is only accurately quantified by a tight bound, the absence of which renders the direct robustness comparison between different quantum classifiers implausible. Practically, an optimal robustness certification protocol, in the sense of being capable of faithfully reporting the noise tolerance and resilience of a quantum algorithm can only arise from a robustness condition which is both sufficient and necessary. Here we set out to confront both aspects of this open problem by generalising the state-of-the-art classical wisdom on certifiable adversarial robustness into the quantum realm.
    
    The pressing demand for robustness against adversarial attacks is arguably even more self-evident under the classical setting in the present era of wide-spread industrial adaptation of machine learning~\cite{szegedy2014,goodfellow2015,eykholt2018}. Many heuristic defence strategies have been proposed but have subsequently been shown to fail against suitably powerful adversaries~\cite{carlini2017,athalye2018}. 
    In response, provable defence mechanisms that provide robustness guarantees have been developed. One line of work, interval bound propagation, uses interval arithmetic~\cite{gowal2018,mirman2018} to certify neural networks.
    Another approach makes use of randomizing inputs and adopts techniques from differential privacy~\cite{lecuyer2019} and, to our particular interest, statistical hypothesis testing~\cite{cohen2019certified,weber2020rab} which has a natural counter-part in the quantum domain. Since the pioneering works by Helstrom~\cite{helstrom67} and Holevo~\cite{holevo1973}, the task of quantum hypothesis testing (QHT) has been well-studied and regarded as one of the foundational tasks in quantum information, with profound linkages with topics ranging from quantum communication~\cite{wang2012,matthews2014}, estimation theory~\cite{helstrom76}, to quantum illumination~\cite{wilde2017,lloyd2008illumination}.
    
    In this work, we lay bare a fundamental connection between quantum hypothesis testing and the robustness of quantum classifiers against unknown noise sources.
    The methods of QHT enable us to derive a robustness condition which, in contrast to other methods, is both \emph{sufficient and necessary} and puts constraints on the amount of noise that a classifier can tolerate. 
    Due to tightness, these constraints allow for an accurate description of noise-tolerance. Absence of tightness, on the other hand, would underestimate the true degree of such noise tolerance.
    Based on these theoretical findings, we provide (1) an optimal robustness certification protocol to assess the degree of tolerance against input perturbations (independent of whether these occur due to natural or adversarial noise), (2) a protocol to verify whether classifying a perturbed (noisy) input has had the same outcome as classifying the clean (noiseless) input, without requiring access to the latter, and (3) tight robustness conditions on parameters for amplitude and phase damping noise.
    In addition, we will also consider randomizing quantum inputs, what can be seen as
    a quantum generalisation to randomized smoothing, a technique that has recently been applied to certify the robustness of classical machine learning models~\cite{cohen2019certified}.
    The conceptual foundation of our approach is rooted in the inherently probabilistic nature of quantum classifiers. Intuitively,
    while QHT is concerned with the question of how to optimally discriminate between two given states, certifying adversarial robustness aims at giving a guarantee for which two states can \emph{not} be discriminated. These two seemingly contrasting notions go hand in hand and, as we will see, give rise to optimal robustness conditions fully expressible in the language of QHT.  Furthermore, while we focus on robustness in a worst-case scenario, our results naturally cover narrower classes of \emph{known} noise sources and can potentially be put in context with other areas such as error mitigation and error tolerance in the NISQ era.
    Finally, while we treat robustness in the context of quantum machine learning, our results in principle do not require the decision function to be learned from data. 
    Rather, our results naturally cover a larger class of quantum algorithms whose outcomes are determined by the most likely measurement outcome. Our robustness conditions on quantum states are then simply conditions under which the given measurement outcome remains the most likely outcome.

    The remainder of this paper is organized as follows: 
    We first introduce the notations and terminologies and review results from QHT essential for our purpose. 
    We then proceed to formally define quantum classifiers and the assumptions on the threat model. 
    In `Results', we present our main findings on provable robustness from quantum hypothesis testing.
    Additionally, these results are demonstrated and visualised with a simple toy example for which we also consider the randomized input setting and analyse specifically randomization with depolarization channel. In `Discussion' we conclude with a higher-level view on our findings and layout several related open problems with an outlook for future research.
    Finally, in `Methods', we give proofs for central results: the robustness condition in terms of type-II error probabilities of QHT, the tightness of this result and, finally, the method used to derive robustness conditions in terms of {fidelity}.

\begin{table*}[t]
    \centering
    \resizebox{\textwidth}{!}{
    \begin{threeparttable}
	\caption{Summary of Results. In this work, we establish a fundamental connection between QHT and the robustness of quantum classification algorithms against adversarial input perturbations. This connection naturally leads to a robustness condition {formulated as a semi-definite program} in terms of optimal type-II error probabilities of distinguishing between benign and adversarial states (QHT condition: Theorem \ref{thm:main}). 
    Under certain practical assumptions about the class probabilities on benign input, we prove that the QHT condition is optimal (Theorem \ref{thm:tightness}). 
    We then show that the QHT condition implies closed form solutions in terms of explicit robustness bound on the fidelity, Bures metric and trace distance.
    We numerically compare an alternative robustness bound directly implied by the definition of trace distance and application of H\"older duality (Lemma \ref{lem: elementary bound} \& Ref. \cite{guan2020robustness}) with the explicit forms of the robustness bounds arising from QHT (FIG. \ref{fig:bounds_contour}).
    Based on these technical findings, we provide a practical protocol to asses the resilience of a classifier against adversarial perturbations, a protocol to certify whether a given noisy input has been classified the same as the noiseless input, without requiring access to the latter, and we derive robustness bounds on noise parameters in amplitude and phase damping.
    Finally, we instantiate our results with a single-qubit pure state example both in the noiseless and depolarization smoothing input scenarios, which allows for numerical comparison of all the known robustness bounds, arising from H\"older duality, differential privacy~\cite{du2020quantum} and QHT (FIG.~\ref{fig:depolarization_bounds}).
    Tight robustness conditions are indicated in \textbf{bold font} in the table.}
	\label{table: results}
		\centering
        \begin{tabular}{ c c c c c c c c }
			\toprule
		    & \multirow{2}[1]{*}{Input States} & \multirow{2}[1]{*}{\makecell{Quantum\\Differential Privacy}} & \multirow{2}[1]{*}{\makecell{Hölder\\Duality}} & \multicolumn{4}{c}{Quantum Hypothesis Testing}\\
			\cmidrule(lr){5-8}
			\multicolumn{4}{c}{} & SDP\tnote{a} & Fidelity & Bures Metric & Trace Distance\\
			\toprule
			\multirow{2}[2]{*}{No Smoothing} & Pure & \multirow{2}[2]{*}{---} &  \multirow{2}[2]{*}{Lemma~\ref{lem: elementary bound}\tnote{b}\hspace{0.5em}} & \multirow{2}[2]{*}{\bf Theorem~\ref{thm:main}} & \multirow{2}[2]{*}{\bf Theorem~\ref{thm:fidelity_bound}} & \multirow{2}[2]{*}{\bf Eq.~(\ref{eq:bures_bound_mixed})} & \bf \bf Eq.~(\ref{eq:trace_bound_pure})\\
			\cmidrule(lr){2-2}\cmidrule(lr){8-8}
			& Mixed & & & & & & Lemma~\ref{lem: elementary bound}\\
			\midrule
			\multirow{2}[2]{*}{\makecell{Depolarization\\Smoothing}} & Pure & \multirow{2}[2]{*}{Lemma~2 in Ref.~\cite{du2020quantum}} & \multirow{2}[2]{*}{Eq.~(\ref{eq:holder_bound_noisy})} & \multirow{2}[2]{*}{\bf Theorem~\ref{thm:main}} & \multirow{2}[2]{*}{---} & \multirow{2}[2]{*}{---} & {\bf Eq.~(\ref{eq:depolarization_bound_main}) (single-qubit)}\\
			\cmidrule(lr){2-2}\cmidrule(lr){8-8}
			& Mixed & & & & & & ---\\
			\bottomrule
	    \end{tabular}
	    \begin{tablenotes}
	        \item[a] Robustness condition expressed in terms of type-II error probabilities $\beta^*$ associated with an optimal quantum hypothesis test.
	        \item[b] Independently discovered in Ref.~\cite{guan2020robustness}.
	    \end{tablenotes}
	    \end{threeparttable}}
\end{table*}

\section{Results}
\label{sec:main_results}

\subsection{Preliminaries}

{\it Notation.} Let $\cH$ be a Hilbert space of finite dimension $d:=\mathrm{dim}(\cH) < \infty$ corresponding to the quantum system of interest.
The space of linear operators acting on $\cH$ is denoted by $\cL(\cH)$ and the identity operator on $\cH$ is written as $\Id$. If not clear from context, the dimensionality is explicitly indicated through the notation $\Id_{d}$.
The set of density operators (i.e.~positive semi-definite trace-one Hermitian matrices) acting on $\cH$, is denoted by $\cS(\cH)$ and elements of $\cS(\cH)$ are written in lowercase Greek letters. The Dirac notation will be adopted whereby Hilbert space vectors are written as $\ket{\psi}$ and their dual as $\bra{\psi}$.
We will use the terminology density operator and quantum state interchangeably.
For two Hermitian operators $A,\,B\in\cL(\cH)$ we write $A > B$ ($A \geq B$) if $A-B$ is positive (semi-)definite and $A < B$ ($A \leq B$) if $A-B$ is negative (semi-)definite. 
For a Hermitian operator $A\in\cL(\cH)$. with spectral decomposition $A = \sum_{i}\lambda_i P_i$, we write $\{A > 0\}:=\sum_{i\colon\lambda_i > 0}P_i$ (and analogously $\{A < 0\}:=\sum_{i\colon\lambda_i < 0}P_i$) for the projection onto the eigenspace of $A$ associated with positive (negative) eigenvalues. 
The Hermitian transpose of an operator $A$ is written as $A^\dagger$ and the complex conjugate of a complex number $z\in\bC$ as $\Bar{z}$.
For two density operators $\rho$ and $\sigma$, the trace distance is defined as $T(\rho,\,\sigma):=\frac{1}{2}\norm[1]{\rho-\sigma}$ where $\norm[1]{\cdot}$ is the Schatten 1-norm defined on $\cL(\cH)$ and given by $\norm[1]{A}:=\Tr{\abs{A}}$ with $\abs{A}=\sqrt{A^\dagger A}$. 
The {Uhlmann} fidelity between density operators $\rho$ and $\sigma$ is denoted by $F$ and defined as $F(\rho,\,\sigma) := \Tr{\sqrt{\sqrt{\rho}\sigma\sqrt{\rho}}}^2$ which for pure states reduces to the squared overlap $F(\ket{\psi},\,\ket{\phi}) = \abs{\braket{\psi}{\phi}}^2$. 
Finally, the Bures metric is denoted by $d_{\mathrm{B}}$ and is closely related to the Uhlmann fidelity via $d_{\mathrm{B}}(\rho,\,\sigma) = [2(1-\sqrt{F(\rho,\,\sigma)})]^\frac{1}{2}$.\\

{\it Quantum Hypothesis Testing.} Typically, QHT is formulated in terms of state discrimination where several quantum states have to be discriminated through a measurement~\cite{helstrom67}.
In binary quantum hypothesis testing, the aim is to decide whether a given unknown quantum system is in one of two states corresponding to the null and alternative hypothesis.
Any such test is represented by an operator $0\leq M \leq \Id_{d}$, which corresponds to rejecting the null in favor of the alternative.
The two central quantities of interest are the probabilities of making a type-I or type-II error. The former corresponds to rejecting the null when it is true, while the latter occurs if the null is accepted when the alternative is true. Specifically, for density operators $\sigma\in\cS(\cH)$ and $\rho\in\cS(\cH)$ describing the null and alternative hypothesis, the type-I error probability is defined as $\alpha(M)$ 
and the 
type-II error probability as $\beta(M)$, so that
\begin{align}
    \alpha(M;\,\sigma) &:= \Tr{\sigma M}\ws\ws &\text{(type-I error)}\\
    \beta(M;\,\rho) &:= \Tr{\rho(\Id - M)} &\text{(type-II error)}
\end{align}
In the Bayesian setting, the hypotheses $\sigma$ and $\rho$ occur with some prior probabilities $\pi_0$ and $\pi_1$ and  are concerned with finding a test which minimizes the total error probability.
A Bayes optimal test $M$ is one that minimizes the posterior probability $\pi_0\cdot\alpha(M) + \pi_1\cdot\beta(M)$.

In this paper, we consider \emph{asymmetric} hypothesis testing {(Neyman-Pearson approach)}~\cite{helstrom76}, 
where the two types of errors are associated with a different cost. Given a maximal allowed probability for the type I error, the goal is to minimize the probability of the type II error. Specifically, one aims to solve the {semidefinite program (SDP)}
\begin{equation}
    \label{eq:type2_error_sdp}
    \begin{aligned}
        \beta^*_{\alpha_0}(\sigma,\,\rho) := \mathrm{minimize}\ws&\beta(M;\,\rho)\\
        \mathrm{s.t.}\ws&\alpha(M;\,\sigma) \leq \alpha_0,\\
        \ws & 0 \leq M \leq \Id_{d}
    \end{aligned}
\end{equation}
Optimal tests can be expressed in terms of projections onto the eigenspaces of the operator $\rho - t\sigma$ where $t$ is a non-negative number. More specifically, for $t\geq 0$ let $P_{t,+}:=\{\rho-t\sigma>0\}$, $P_{t,-}:=\{\rho-t\sigma<0\}$ and $P_{t,0}:=\Id - P_{t,+} - P_{t,-}$ be the projections onto the eigenspaces of $\rho -t\sigma$ associated with positive, negative and zero eigenvalues.
The quantum analogue to the Neyman-Pearson Lemma \cite{neyman1933} shows optimality of operators of the form
\begin{equation}
    \label{eq:np_operators}
    M_t := P_{t,+} + X_t,\ws 0 \leq X_t \leq P_{t,0}.
\end{equation}
The choice of the scalar $t\geq0$ and the operator $X_t$ is such that the preassigned type-I error probability $\alpha_0$ is attained.
An explicit construction for these operators is based on the inequalities
\begin{equation}
    \alpha(P_{\tau(\alpha_0),+}) \leq \alpha_0 \leq \alpha(P_{\tau(\alpha_0),+} + P_{\tau(\alpha_0),0})
\end{equation}
where $\alpha_0\in (0,\,1)$ and $\tau(\alpha_0)$ is the smallest non-negative number such that $\alpha(P_{\tau(\alpha_0),+}) \leq \alpha_0$, i.e. $\tau(\alpha_0) := \inf\{t\geq 0\colon\,\alpha(P_{t,+}) \leq \alpha_0\}$.
These inequalities can be seen from the observation that the function $t \mapsto \alpha(P_{t,+})$ is non-increasing and right-continuous while $t \mapsto \alpha(P_{t,+} + P_{t,0})$ is non-increasing and left-continuous. A detailed proof for this is given in Supplementary Note 1 and 2.
We will henceforth refer to operators of the form~(\ref{eq:np_operators}) as Helstrom operators~\cite{helstrom76}.\\

{\it Quantum classifiers.} We define a $K$-class quantum classifier of states of the quantum system $\cH$, described by density operators, as a map $\cA\colon\cS(\cH)\to\cC$ which maps states $\sigma\in\cS(\cH)$ to class labels $k\in\cC=\{1,\,\ldots,\,K\}$.
{Any such classifier is described by a completely positive and trace-preserving (CPTP) map $\cE$ and a POVM $\{\Pi_{k}\}_{k}$.
Formally, a quantum state $\sigma$ is passed through the quantum channel $\cE$ and then the measurement $\{\Pi_{k}\}_{k}$ is performed. 
Finally, the probability of measuring outcome $k$ is identified with the class probability $\y_{k}(\sigma)$, i.e.}
\begin{equation}
    \label{eq:classifier}
    \sigma \mapsto \y_{k}(\sigma) := \Tr{\Pi_{k}\cE(\sigma)}.
\end{equation}
We treat the positive-operator valued measure (POVM) element $\Pi_{k}$ as a projector $\Pi_{k} = \ket{k}\bra{k}\tens\Id_{d/K}$ which determines whether the output is classified into class $k$. This can be done without loss of generality by Naimark's dilation since $\cE$ is kept arbitrary and potentially involves ancillary qubits and a general POVM element can be expressed as a projector on the larger Hilbert space. The final prediction is given by the most likely class 
\begin{align}
    \cA(\sigma)\equiv\arg\max_{k}\y_{k}(\sigma). 
\end{align}
{Throughout this paper, we refer to $\cA$ as the \emph{classifier} and to $\y$ as the \emph{score function}.}
In the context of quantum machine learning, the input state $\sigma$ can be an encoding of classical data by means of, for example, amplitude encoding or otherwise \cite{larose2020robust, zhao2018smooth}, or inherently quantum input data, while $\cE$ can be realized, for example, by a trained parametrized quantum circuit potentially involving ancillary registers \cite{benedetti2019parameterized}. However, it is worth noting that the above-defined notion of quantum classifier more generally describes the procedure of a broader class of quantum algorithms whose output is obtained by repeated sampling of measurement outcomes.\\

\begin{figure}[t]
    \centering
    \includegraphics[width=\linewidth]{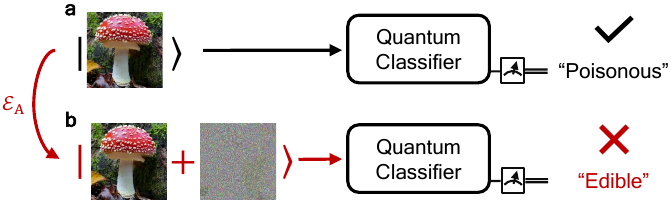}
    \caption{{Adversarial attack. \textbf{a} A quantum classifier correctly classifies the (toxic) mushroom as ``poisonous''. \textbf{b} An adversary perturbs the image to fool the classifier into believing that the mushroom is ``edible''.}}
    \label{fig:mushroom}
\end{figure}

{\it Quantum adversarial robustness.} Adversarial examples are attacks on classification models where an adversary aims to induce a misclassification using typically imperceptible modifications of a benign input example. 
Specifically, given a classifier $\cA$ and a benign input state $\sigma$, an adversary can craft a small perturbation $\sigma \to \rho$ which results in a misclassification, i.e. $\cA(\rho) \neq \cA(\sigma)$. 
An illustration for this threat scenario is given in FIG.~\ref{fig:mushroom}.
In this paper, we seek a worst-case robustness guarantee against \emph{any} possible attack: as long as $\rho$ does not differ from $\sigma$ by more than a certain amount, then it is guaranteed that $\cA(\sigma) = \cA(\rho)$ independently of how the adversarial state $\rho$ has been crafted. 
Formally, suppose the quantum classifier $\cA$ takes as input a \emph{benign} quantum state $\sigma\in\cS(\cH)$ and produces a measurement outcome denoted by the class $k\in\cC$ with probability $\y_{k}(\sigma) = \Tr{\Pi_{k}\cE(\sigma)}$. 
Recall that the prediction of $\cA$ is taken to be the most likely class $k_{\mathrm{A}} = \arg\max_k\y_{k}(\sigma)$. 
An adversary aims to alter the output probability distribution so as to change the most likely class by applying an arbitrary quantum operation $\cE_{\mathrm{A}}\colon\cS(\cH)\to\cS(\cH)$ to $\sigma$ resulting in the \emph{adversarial state} $\rho = \cE_{\mathrm{A}}(\sigma)$.
Finally, we say that the classifier $\y$ is provably robust around $\sigma$ with respect to the robustness condition $\cR$, if for any $\rho$ which satisfies $\cR$, it is guaranteed that
$\cA(\rho) = \cA(\sigma)$.

In the following, we will derive a robustness condition for quantum classifiers with the QHT formalism, which provides a provable guarantee for the outcome of a computation being unaffected by the worst-case input noise or perturbation under a given set of constraints. 
In the regime where the most likely class is measured with probability lower bounded by $p_{\mathrm{A}} > 1/2$ and the runner up class is less likely than $p_{\mathrm{B}} = 1-p_{\mathrm{A}}$, we prove tightness of the robustness bound, hence demonstrating that the QHT condition is at least partially optimal.
The QHT robustness condition, in its full generality, has an SDP formulation 
in terms of the {optimal} type-II error probabilities.
We then simplify this condition and derive closed form solutions in terms of Uhlmann fidelity, Bures metric and trace distance between benign and adversarial inputs.
The closed form solutions in terms of fidelity and Bures metric are shown to be sufficient and necessary for general states and in the same regime where the SDP formulation is proven to be tight. In the case of trace distance, this can be claimed for pure states, while the bound for mixed states occurs to be weaker.
These results stemming from QHT considerations are then contrasted and compared with an alternative approach which directly applies H\"older duality to trace distances to obtain a sufficient robustness condition.
The different robustness bounds and robustness conditions are summarized in Table~\ref{table: results}.

\subsection{Robustness condition from quantum hypothesis testing}
Recall that quantum hypothesis testing is concerned with the question of finding measurements that optimally discriminate between two states. A measurement is said to be optimal if it minimizes the probabilities of identifying the quantum system to be in the state $\sigma$, corresponding to the null hypothesis, when in fact it is in the alternative state $\rho$, and vice versa. When considering provable robustness, on the other hand, one aims to find a neighbourhood around a benign state $\sigma$ where the class which is most likely to be measured is constant or, expressed differently, where the classifier can not discriminate between states. It becomes thus clear that quantum hypothesis testing and classification robustness aim to achieve a similar goal, although viewed from different angles. {Indeed, as it turns out, QHT determines the robust region around $\sigma$ to be the set of states (i.e. alternative hypotheses) for which the optimal type-II error probability $\beta^*$ is larger than $1/2$.}

To establish this connection more formally, we identify the benign state with the null hypothesis $\sigma$ and the adversarial state with the alternative $\rho$.
We note that, in the Heisenberg picture, we can identify a the score function $\y$ of a classifier $\cA$ with a POVM $\{\Pi_{k}\}_{k}$.
For $k_{\mathrm{A}}=\cA(\sigma)$, the operator $\Id-\Pi_{k_{\mathrm{A}}}$ (and thus the the classifier $\cA$) can be viewed as a hypothesis test discriminating between $\sigma$ and $\rho$.
{Notice that, for $p_{\mathrm{A}}\in[0,\,1]$ with $\y_{k_{\mathrm{A}}}(\sigma) = \Tr{\Pi_{k_{\mathrm{A}}}\sigma} \geq p_{\mathrm{A}}$, the operator $\Id_{d} - \Pi_{k_{\mathrm{A}}}$ is feasible for the SDP $\beta^*_{1-p_{\mathrm{A}}}(\sigma,\,\rho)$ in~\eqref{eq:type2_error_sdp} and hence
\begin{equation}
    \y_{k_{\mathrm{A}}}(\rho) = \beta(\Id_{d} - \Pi_{k_{\mathrm{A}}};\,\rho) \geq \beta^*_{1-p_{\mathrm{A}}}(\sigma,\,\rho).
\end{equation}
Thus, it is guaranteed that $k_{\mathrm{A}} = \cA(\rho)$ for any $\rho$ with $\beta^*_{1-p_{\mathrm{A}}}(\sigma,\,\rho) > 1/2$.}
The following theorem makes this reasoning concise and extends to the setting where the probability of measuring the second most likely class is upper-bounded by $p_{\mathrm{B}}$.
\begin{thm}[QHT robustness bound]
    \label{thm:main}
    Let $\sigma,\,\rho\in\cS(\cH)$ be benign and adversarial quantum states and let $\cA$ be a quantum classifier with score function $\y$. Suppose that for $k_{\mathrm{A}}\in\cC$ and $p_{\mathrm{A}},\,p_{\mathrm{B}}\in[0,\,1]$, the score function $\y$ satisfies
    \begin{align}
        \label{eq:class_probs}
        \y_{k_{\mathrm{A}}}(\sigma) \geq p_{\mathrm{A}} > p_{\mathrm{B}} \geq \max_{k \neq k_{\mathrm{A}}} \y_{k}(\sigma).
    \end{align}
    Then, it is guaranteed that $\cA(\rho) = \cA(\sigma)$ for any $\rho$ with
    \begin{equation}
        \label{eq:robustness_condition}
        \beta^*_{1-p_{\mathrm{A}}}(\sigma,\,\rho)+\beta^*_{p_{\mathrm{B}}}(\sigma,\,\rho) > 1 
    \end{equation}
\end{thm}

To get some more intuition of Theorem \ref{thm:main}, we first note that for $p_{\mathrm{B}} = 1-p_{\mathrm{A}}$, the robustness condition~(\ref{eq:robustness_condition}) simplifies to
\begin{equation}
    \beta^*_{1-p_{\mathrm{A}}}(\sigma,\,\rho) > 1/2
\end{equation}
With this, the relation between quantum hypothesis testing and robustness becomes more evident: if the \emph{optimal} hypothesis test performs poorly when discriminating the two states, then a classifier will predict both states to belong to the same class. In other words, viewing a classifier as a hypothesis test between the benign input $\sigma$ and the adversarial $\rho$, the optimality of the Helstrom operators implies that the classifier $\y$ is a worse discriminator and will also not distinguish the states, or, phrased differently, it is robust. This result formalizes the intuitive connection between quantum hypothesis testing and robustness of quantum classifiers. While the former is concerned with finding operators that are optimal for discriminating two states, the latter is concerned with finding conditions on states for which a classifier does \emph{not} discriminate.\\

{\it Optimality.} The robustness condition \eqref{eq:robustness_condition} from QHT is provably optimal in the regime of $p_{\mathrm{A}}+p_{\mathrm{B}}=1$, which covers binary classifications in full generality and multi-class classification where the most likely class is measured with probability larger than $p_{\mathrm{A}} > \frac{1}{2}$. The robustness condition is tight in the sense that, whenever condition~\eqref{eq:robustness_condition} is violated, then there exists a classifier $\cA^\star$ which is consistent with the class probabilities~\eqref{eq:class_probs} on the benign input but which will classify the adversarial input differently from the benign input.
The following theorem demonstrates this notion of tightness by explicitly constructing the worst-case classifier $\cA^\star$.  
\begin{thm}[Tightness]
    \label{thm:tightness}
    Suppose that $p_{\mathrm{A}} + p_{\mathrm{B}} = 1$. Then, if the adversarial state $\rho$ violates condition \eqref{eq:robustness_condition},
    there exists a quantum classifier $\cA^\star$ that is consistent with the class probabilities~\eqref{eq:class_probs} and for which $\cA^\star(\rho) \neq \cA^\star(\sigma)$.
\end{thm}
The main idea of the proof relies on the explicit construction of a ``worst-case'' classifier with Helstrom operators and which classifies $\rho$ differently from $\sigma$ while still being consistent with the class probabilities~(\ref{eq:class_probs}).
We refer the reader to `Methods' for a detailed proof.
Whether or not the QHT robustness condition is tight for $p_{\mathrm{A}}+p_{\mathrm{B}}<1$ is an interesting open question for future research. It turns out that a worst-case classifier which is consistent with $p_{\mathrm{A}}$ and $p_{\mathrm{B}}$ for benign input but leads to a different classification on adversarial input upon violating condition \eqref{eq:robustness_condition}, if exists, is more challenging to construct for these cases. If such a tightness result for all class probability regimes would be proven, there would be a complete characterization for the robustness of quantum classifiers.

\subsection{Closed form robustness conditions}
\label{subsec:explicit_bounds}
Although Theorem \ref{thm:main} provides a general condition for robustness with provable tightness, it is formulated as a semidefinite program in terms of type-II error probabilities of QHT. 
To get a more intuitive and operationally convenient perspective, we wish to derive a condition for robustness in terms of a meaningful notion of difference between quantum states. 
Specifically, based on Theorem~\ref{thm:main}, here we derive robustness conditions expressed in terms of Uhlmann's fidelity $F$, Bures distance $d_{\mathrm{B}}$ and in terms of the trace distance $T$.
To that end, we first concentrate on pure state inputs and will then leverage these bounds to mixed states.
Finally, we show that expressing robustness in terms of fidelity or Bures distance results in a tight bound for both pure and mixed states, while for trace distance the same can only be claimed in the case of pure states.\\

{\it Pure states.} We first assume that both the benign and the adversarial states are pure. This assumption allows us to first write the optimal type-II error probabilities $\beta^*_\alpha(\rho,\,\sigma)$ as a function of $\alpha$ and the fidelity between $\rho$ and $\sigma$. This leads to a robustness bound on the fidelity and subsequently to a bound on the trace distance and on the Bures distance. Finally, since these conditions are equivalent to the QHT robustness condition~\eqref{eq:robustness_condition}, Theorem~\ref{thm:tightness} implies tightness of these bounds.
\begin{lem}
    \label{lem:pure_state_bound}
    Let $\ket{\psi_\sigma},\,\ket{\psi_\rho}\in\cH$ and let $\cA$ be a quantum classifier. Suppose that for $k_{\mathrm{A}}\in\cC$ and $p_{\mathrm{A}},\,p_{\mathrm{B}}\in[0,\,1]$, we have $k_{\mathrm{A}} = \cA(\psi_\sigma)$ and suppose that the score function $\y$ satisfies~\eqref{eq:class_probs}.
    Then, it is guaranteed that $\cA(\psi_\rho)=\cA(\psi_\sigma)$ for any $\psi_\rho$ with
    \begin{equation}
        \label{eq:fidelity_bound_pure}
        \abs{\braket{\psi_\sigma}{\psi_\rho}}^2 > \frac{1}{2}\left(1 + \sqrt{g(p_{\mathrm{A}},\,p_{\mathrm{B}})}\right),
    \end{equation}
    where the function $g$ is given by
    \begin{equation}
        \label{eq:fidelity_g_function}
        \begin{aligned}
            g(p_{\mathrm{A}},\,p_{\mathrm{B}}) &= 1 - p_{\mathrm{B}} - p_{\mathrm{A}}(1-2p_{\mathrm{B}}) + \\
            &\hspace{4em}2\sqrt{p_{\mathrm{A}} p_{\mathrm{B}}(1-p_{\mathrm{A}})(1-p_{\mathrm{B}})}.
        \end{aligned}
    \end{equation}
    This condition is equivalent to~\eqref{eq:robustness_condition} and is hence both sufficient and necessary whenever $p_{\mathrm{A}} + p_{\mathrm{B}} = 1$.
\end{lem}
This result thus provides a closed form robustness bound which is equivalent to the SDP formulation in condition~\eqref{eq:robustness_condition} and is hence sufficient and necessary in the regime $p_{\mathrm{A}} + p_{\mathrm{B}} = 1$. We remark that, under this assumption, the robustness bound~\eqref{eq:fidelity_bound_pure} has the compact form
\begin{equation}
    \abs{\braket{\psi_\sigma}{\psi_\rho}}^2 > \frac{1}{2} + \sqrt{p_{\mathrm{A}}(1-p_{\mathrm{A}})}.
\end{equation}
Due to its relation with the Uhlmann fidelity, it is straight forward to obtain a robustness condition in terms of Bures metric. Namely, the condition
\begin{equation}
    \label{eq:bures_bound_pure}
    d_{\mathrm{B}}(\ket{\psi_\rho},\,\ket{\psi_\sigma}) < \left[2 - \sqrt{2(1+\sqrt{g(p_{\mathrm{A}},\,p_{\mathrm{B}})})}\right]^\frac{1}{2}
\end{equation}
is equivalent to~\eqref{eq:robustness_condition}. 
Furthermore, since the states are pure, we can directly link~\eqref{eq:fidelity_bound_pure} to a bound in terms of the trace distance via the relation $T(\ket{\psi_\rho},\,\ket{\psi_\sigma})^2 = 1 - \abs{\braket{\psi_\sigma}{\psi_\rho}}^2$, so that
\begin{equation}
    \label{eq:trace_bound_pure}
    T(\ket{\psi_\rho},\,\ket{\psi_\sigma}) < \left[\frac{1}{2}\left(1 - \sqrt{g(p_{\mathrm{A}},\,p_{\mathrm{B}})}\right)\right]^\frac{1}{2}
\end{equation}
is equivalent to~\eqref{eq:robustness_condition}. Due to the equivalence of these bounds to~\eqref{eq:robustness_condition}, Theorem~\ref{thm:tightness} applies and it follows that both bounds are sufficient and necessary in the regime where $p_{\mathrm{A}} + p_{\mathrm{B}} = 1$. In the following, we will extend these results to mixed states and show that both the fidelity and Bures metric bounds are tight.\\

{\it Mixed states.} Reasoning about the robustness of a classifier if the input states are mixed, rather than just for pure states, is practically relevant for a number of reasons. Firstly, in a realistic scenario, the assumption that an adversary can only produce pure states is too restrictive and gives an incomplete picture.
Secondly, if we wish to reason about the resilience of a classifier against a given noise model (e.g. amplitude damping), then the robustness condition needs to be valid for mixed states as these noise models typically produce mixed states.
Finally, in the case where we wish to certify whether a classification on a noisy input has had the same outcome as on the noiseless input, a robustness condition for mixed states is also required.
For these reasons, and having established closed form robustness bounds which are both sufficient and necessary for pure states, here we aim to extend these results to the mixed state setting. The following theorem extends the fidelity bound~\eqref{eq:fidelity_bound_pure} for mixed states. As for pure states, it is then straight forward to obtain a bound in terms of the Bures metric.
\begin{thm}
    \label{thm:fidelity_bound}
    Let $\sigma,\,\rho\in\cS(\cH)$ and let $\cA$ be a quantum classifier. Suppose that for $k_{\mathrm{A}}\in\cC$ and $p_{\mathrm{A}},\,p_{\mathrm{B}}\in[0,\,1]$, we have $k_{\mathrm{A}} = \cA(\sigma)$ and suppose that the score function $\y$ satisfies~\eqref{eq:class_probs}.
    Then, it is guaranteed that $\cA(\rho)=\cA(\sigma)$ for any $\rho$ with
    \begin{equation}
        \label{eq:robustness_bound_fidelity}
        F(\rho,\,\sigma) > \frac{1}{2}\left(1 + \sqrt{g(p_{\mathrm{A}},\,p_{\mathrm{B}})}\right) =: r_{\mathrm{F}}
    \end{equation}
    where $g$ is defined as in~\eqref{eq:fidelity_g_function}. This condition is both sufficient and necessary if $p_{\mathrm{A}} + p_{\mathrm{B}} = 1$.
\end{thm}
\begin{proof}
    To show sufficiency of~\eqref{eq:robustness_bound_fidelity}, we notice that $\y$ can be rewritten as
    \begin{align}
        \label{eq:classifier_purified}
        \y_{k}(\sigma) &= \Tr{\Pi_{k}\cE(\sigma)}\\
        &= \Tr{\Pi_{k}(\cE\circ\mathrm{Tr}_{\mathrm{E}})(\ketbra{\psi_\sigma}{\psi_\sigma})}
    \end{align}
    where $\ket{\psi_\sigma}$ is a purification of $\sigma$ with purifying system $\mathrm{E}$ and $\mathrm{Tr}_{\mathrm{E}}$ denotes the partial trace over $\mathrm{E}$. We can thus view $\y$ as a score function on the larger Hilbert space which admits the same class probabilities for $\sigma$ and any purification of $\sigma$ (and equally for $\rho$).
    It follows from Uhlmann's Theorem that there exist purifications $\ket{\psi_\sigma}$ and $\ket{\psi_\rho}$ such that $F(\rho,\,\sigma) = \abs{\braket{\psi_\sigma}{\psi_\rho}}^2$. Robustness at $\rho$ then follows from~\eqref{eq:robustness_bound_fidelity} by~\eqref{eq:classifier_purified} and Lemma~\ref{lem:pure_state_bound}.
    To see that the bound is necessary when $p_{\mathrm{A}} + p_{\mathrm{B}} = 1$, suppose that there exists some $\Tilde{r}_{\mathrm{F}} < r_{\mathrm{F}}$ such that $F(\sigma,\,\rho) > \Tilde{r}_{\mathrm{F}}$ implies that $\cA(\rho) = \cA(\sigma)$. Since pure states are a subset of mixed states, this bound must also hold for pure states. In particular, suppose $\ket{\psi_\rho}$ is such that $\Tilde{r}_{\mathrm{F}} < \abs{\braket{\psi_\rho}{\psi_\sigma}}^2 \leq r_{\mathrm{F}}$. However, this is a contradiction, since $\abs{\braket{\psi_\rho}{\psi_\sigma}}^2 \geq r_{\mathrm{F}}$ is both sufficient and necessary in the given regime, i.e. by Theorem~\ref{thm:tightness}, there exists a classifier $\cA^\star$ whose score function satisfies~\eqref{eq:class_probs} and for which $\cA^\star(\psi_\sigma) \neq \cA^\star(\psi_\rho)$. It follows that $\Tilde{r}_{\mathrm{F}} \geq r_{\mathrm{F}}$ and hence the claim of the theorem.
\end{proof}
Due to the close relation between Uhlmann fidelity and the Bures metric, we arrive at a robustness condition for mixed states in terms of $d_{\mathrm{B}}$, namely
\begin{equation}
    \label{eq:bures_bound_mixed}
    d_{\mathrm{B}}(\rho,\,\sigma) < \left[2 - \sqrt{2(1+\sqrt{g(p_{\mathrm{A}},\,p_{\mathrm{B}})})}\right]^\frac{1}{2}
\end{equation}
which inherits the tightness properties of the fidelity bound~\eqref{eq:robustness_bound_fidelity}.
In contrast to the pure state case, here it is less straight forward to obtain a robustness bound in terms of trace distance. However, we can still build on Lemma~\ref{lem:pure_state_bound} and the trace distance bound for pure states~\eqref{eq:trace_bound_pure} to obtain a sufficient robustness condition. Namely, when assuming that the benign state is pure, but the adversarial state is allowed to be mixed we have the following result.
\begin{cor}[Pure Benign \& Mixed Adversarial States]
    \label{cor:trace_bound_pure_mixed}
    Let $\sigma,\,\rho\in\cS(\cH)$ and suppose that $\sigma = \ketbra{\psi_\sigma}{\psi_\sigma}$ is pure. Let $\cA$ be a quantum classifier and suppose that for $k_{\mathrm{A}}\in\cC$ and $p_{\mathrm{A}},\,p_{\mathrm{B}}\in[0,\,1]$, we have $k_{\mathrm{A}} = \cA(\sigma)$ and suppose that the score function $\y$ satisfies~\eqref{eq:class_probs}.
    Then, it is guaranteed that $\cA(\rho)=\cA(\sigma)$ for any $\rho$ with
    \begin{align}
        \label{robustcondition:puremixed}
        T(\rho,\,\sigma) < \delta(p_{\mathrm{A}},\,p_{\mathrm{B}})\left(1-\sqrt{1-\delta(p_{\mathrm{A}},\,p_{\mathrm{B}})^2}\right)
    \end{align}
    where 
    $
    \delta(p_{\mathrm{A}},\,p_{\mathrm{B}}) =[\frac{1}{2}\left(1-g(p_{\mathrm{A}}, p_{\mathrm{B}})\right)]^\frac{1}{2}.
    $
\end{cor}
We refer the reader to supplementary note 4 for a detailed proof of this result.
Intuitively, condition \eqref{robustcondition:puremixed} is derived by noting that any convex mixture of robust pure states must also be robust, thus membership of the set of mixed states enclosed by the convex hull of robust pure states (certified by {equation~\eqref{eq:trace_bound_pure}}
is a natural sufficient condition for robustness. As such, the corresponding robustness radius in condition~\eqref{robustcondition:puremixed} is obtained by lower-bounding, with triangle inequalities, the radius of the maximal sphere centered at $\sigma$ within the convex hull. However, the generalization from
Lemma~\ref{lem:pure_state_bound} and equation~\eqref{eq:trace_bound_pure}
to Corollary \ref{cor:trace_bound_pure_mixed}, mediated by the above geometrical argument, results in a sacrifice of tightness. How or to what extent such loosening of the explicit bound in the cases of mixed states may be avoided or ameliorated remains an open question. In the following, we compare the trace distance bounds from QHT with a robustness condition derived from an entirely different technique. 

We note that a sufficient condition can be obtained from a somewhat straightforward application of H\"{o}lder duality for trace norms:
\begin{lem}[H\"older duality bound]
    \label{lem: elementary bound}
    Let $\sigma,\,\rho\in\cS(\cH)$ be arbitrary quantum states and let $\cA$ be a quantum classifier. Suppose that for $k_{\mathrm{A}}\in\cC$ and $p_{\mathrm{A}},\,p_{\mathrm{B}}\in[0,\,1]$, we have $k_{\mathrm{A}} = \cA(\sigma)$ and the score function $\y$ satisfies~\eqref{eq:class_probs}.
    Then, it is guaranteed that $\cA(\rho)=\cA(\sigma)$ for any $\rho$ with
    \begin{equation}
        \label{eq:elementary bound}
        \frac{1}{2}\left\|\rho - \sigma\right\|_1 < \frac{p_{\mathrm{A}} - p_{\mathrm{B}}}{2}.
    \end{equation}
\end{lem}
\begin{proof}
Let $\delta:=\frac{1}{2}\|\rho-\sigma\|_1=\sup_{0\le P \le \mathbb{I}}\Tr{P(\rho-\sigma)}$, which follows from H\"{o}lder duality. We have that $\y_{k_{\mathrm{A}}}(\sigma)-\y_{k_{\mathrm{A}}}(\rho)\le \delta$ and that $\y_{k_{\mathrm{A}}}(\sigma) \geq p_{\mathrm{A}}$, hence $\y_{k_{\mathrm{A}}}(\rho)\ge p_{\mathrm{A}}-\delta$. We also have, for $k'$ such that $\y_{\mathrm{k'}}(\rho)=\max_{k \neq k_{\mathrm{A}}} \y_{k}(\rho)$, that $\y_{\mathrm{k'}}(\rho)-\y_{\mathrm{k'}}(\sigma)\le \delta$ and that $\y_{\mathrm{k'}}(\sigma)\le p_{\mathrm{B}}$, hence $\max_{k \neq k_{\mathrm{A}}} \y_{k}(\rho)\le p_{\mathrm{B}} + \delta$. Thus $ \frac{1}{2}\left\|\rho - \sigma\right\|_1 < \frac{p_{\mathrm{A}} - p_{\mathrm{B}}}{2} \iff p_{\mathrm{A}}-\delta>p_{\mathrm{B}}+\delta \implies \y_{k_{\mathrm{A}}}(\rho)>\max_{k \neq k_{\mathrm{A}}} \y_{k}(\rho)$.
\end{proof}
We acknowledge the above robustness bound from H\"{o}lder duality was independently discovered in Lemma 1 of Ref. \cite{guan2020robustness}. For intuitive insights, it is worth remarking that the condition \eqref{eq:elementary bound} stems from comparing the maximum probability of distinguishing $\sigma$ and $\rho$ with the optimal measurement (H\"older measurement) with the gap between the first two class probabilities on $\sigma$. Since no classifier can distinguish $\sigma$ and $\rho$ better than the H\"older measurement by definition, \eqref{eq:elementary bound} is clearly a sufficient condition. However, the H\"older measurement on $\sigma$ does not necessarily result in class probabilities consistent with equation~(\ref{eq:class_probs}). Without additional constraints on desired class probabilities on the benign input, the robustness condition \eqref{eq:elementary bound} from H\"{o}lder duality is stronger than necessary. In contrast, the QHT bound from Theorem \ref{thm:main}, albeit implicitly written in the language of hypothesis testing, naturally incorporates such desired constraints. Hence, as expected, this gives rise to a tighter robustness condition.

In summary, the closed form solutions in terms of fidelity and Bures metric completely inherit the tightness of Theorem~\ref{thm:main}, while for trace distance, tightness is inherited for pure states, but partially lost in Corollary~\ref{cor:trace_bound_pure_mixed} for mixed adversarial states.
The numerical comparison between the trace distance bounds from QHT and the H\"{o}lder duality bound is shown in a contour plot in FIG.~\ref{fig:bounds_contour}.

\begin{figure}[t]
    \centering
    \includegraphics[width=\linewidth]{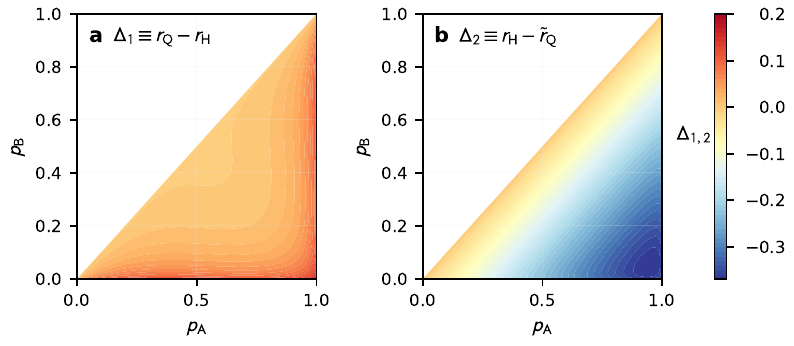}
    \caption{Comparison between robustness bounds in terms of trace distance. \textbf{a} Difference $r_{\mathrm{Q}}-r_{\mathrm{H}}$ between the pure state bound derived from QHT $r_{\mathrm{Q}}$, given in Eq.~\eqref{eq:trace_bound_pure} and the H\"older duality bound $r_{\mathrm{H}}$ from Lemma~\ref{lem: elementary bound}. \textbf{b} Difference $r_{\mathrm{H}} - \Tilde{r}_{\mathrm{Q}}$ between the H\"older duality bound $r_{\mathrm{H}}$ and the bound $\Tilde{r}_{\mathrm{Q}}$ derived from the convex hull approximation to the QHT robustness condition from Theorem~\ref{thm:main} for mixed adversarial states. It can be seen that the pure state bound $r_{\mathrm{Q}}$ is always larger than $r_{\mathrm{H}}$ which in turn is always larger than the convex hull approximation bound $\Tilde{r}_{\mathrm{Q}}$.}
    \label{fig:bounds_contour}
\end{figure}

\subsection{Toy example with single-qubit pure states}
\label{sec:examples_pure_states}

We now present a simple example to highlight the connection between quantum hypothesis testing and classification robustness.
We consider a single-qubit system which is prepared either in the state $\sigma$ or $\rho$ described by
\begin{align}
    \ket{\sigma}&=\ket{0},\\
    \ket{\rho}&=\cos(\theta_0/2)\ket{0} + \sin(\theta_0/2)e^{i\phi_0}\ket{1}
\end{align}
with $\theta_0\in[0,\,\pi)$ and $\phi_0\in[0,\,2\pi)$.
The state $\sigma$ corresponds to the null hypothesis in the QHT setting and to the benign state in the classification setting.
Similarly, $\rho$ corresponds to the alternative hypothesis and adversarial state.
The operators which are central to both QHT and robustness are the Helstrom operators~(\ref{eq:np_operators}) which are derived from the projection operators onto the eigenspaces associated with the non-negative eigenvalues of the operator $\rho - t\sigma$. For this example, the eigenvalues are functions of $t\geq0$ and given by
\begin{align}
    \eta_1 &= \frac{1}{2}(1-t) + R > 0,\\
    \eta_2 &= \frac{1}{2}(1-t) - R \leq 0\\
    R &= \frac{1}{2}\sqrt{(1 - t)^2 + 4t(1-\abs{\gamma}^2)}
\end{align}
where $\gamma$ is the overlap between $\sigma$ and $\rho$ and given by $\gamma = \cos(\theta_0/2)$. For $t > 0$, the Helstrom operators are then given by the projection onto the eigenspace associated with the eigenvalue $\eta_1 > 0$. The projection operator is given by $M_t = \ketbra{\eta_1}{\eta_1}$ with
\begin{align}
        \ket{\eta_1} &= (1 - \eta_1)A_1\ket{0} - \gamma A_1 \ket{\rho}\\
        \abs{A_1}^{-2} &= 2R\abs{\eta_1 - \sin^2(\theta_0/2)}
\end{align}
where $A_1$ is a normalization constant ensuring that $\braket{\eta_1}{\eta_1} = 1$.
Given a preassigned probability $\alpha_0$ for the maximal allowed type-I error probability, we determine $t$ such that $\alpha(M_t) = \alpha_0$.\\

\begin{figure}[t]
    \centering
    \includegraphics[width=.9\linewidth]{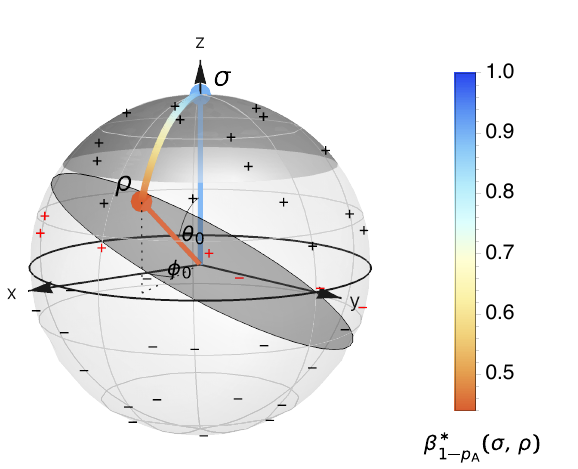}
    \caption{Example classifier for single-qubit quantum states. The decision boundary is represented by the grey disk passing through the origin of the Bloch sphere. The robust region around $\sigma$ is indicated by the dark spherical cap. States belonging to different classes are marked with $+$ and $-$ and are color red if not classified correctly. 
    The colorbar indicates different values for the optimal type-II error probability $\beta^*_{1-p_{\mathrm{A}}}(\sigma,\,\rho)$.
    We see that, for the given classifier, the state $\rho$ is not contained in the robust region around $\sigma$ since the optimal type-II error probability is less than $1/2$ as indicated by the colorbar. The state $\rho$ is thus not guaranteed to be classified correctly by every classifier with the same class probabilities. In the asymmetric hypothesis testing view, an optimal discriminator which admits $0.1$ type-I error probability for testing $\sigma$ against $\rho$ has type-II error probability $0.44$.}
    \label{fig:bloch_sphere}
\end{figure}
{\it Hypothesis testing view.} In QHT, we are given a specific alternative hypothesis $\rho$ and error probability $\alpha_0$ and are interested in finding the minimal type-II error probability. In this example, we pick $\theta_0 = \pi / 3$, $\phi_0=\pi/6$ for the alternative state and set the type-I error probability to $\alpha_0 = 1-p_{\mathrm{A}}= 0.1$. These states are graphically represented on the Bloch sphere in FIG.~\ref{fig:bloch_sphere}. We note that, for this choice of states, we obtain an expression for the eigenvector $\ket{\eta_1}$ given by
\begin{equation}
    \begin{aligned}
        \ket{\eta_1} &= \frac{9 - \sqrt{3}}{\sqrt{30}}\ket{0} - 3\sqrt{\frac{2}{5}}\ket{\rho}.
    \end{aligned}
\end{equation}
which yields the type-II error probability
\begin{equation}
    {\beta^*_{1-p_{\mathrm{A}}}(\sigma,\,\rho) = }\,\beta(M_t) = 1 - \abs{\braket{\eta_1}{\rho}}^2 \approx 0.44 < 1/2.
\end{equation}
We thus see that the optimal hypothesis test can discriminate $\sigma$ and $\rho$ with error probabilities less than $1/2$ since on the Bloch sphere they are located far enough apart. However, since $\beta(M_t) \ngtr 1/2$, Theorem~\ref{thm:main} implies that $\rho$ is not guaranteed to be classified equally as $\sigma$ by a classifier which makes a prediction on $\sigma$ with confidence at least $0.9$. In other words, the two states are far enough apart to be easily discriminated by the optimal hypothesis test but too far apart to be guaranteed to be robust.\\

{\it Classification robustness view.} In this scenario, in contrast to the QHT view, we are not given a specific adversarial state $\rho$, but rather aim to find a condition on a generic $\rho$ such that the classifier is robust for all configurations of $\rho$ that satisfy this condition. Theorem~\ref{thm:main} provides a necessary and sufficient condition for robustness, expressed in terms of $\beta^*$, which, for $p_{\mathrm{B}} = 1-p_{\mathrm{A}}$ and $p_{\mathrm{A}} > 1/2$, reads
\begin{equation}
    \beta^*_{1-p_{\mathrm{A}}}(\sigma,\,\rho) > 1/2
\end{equation}
Recall that the probability and $p_{\mathrm{A}} > 1/2$ is a lower bound to the probability of the most likely class and in this case we set $p_{\mathrm{B}}=1-p_{\mathrm{A}}$ to be the upper bound to the probability of the second most likely class.
For example, as the QHT view shows, for $\alpha_0=1-p_{\mathrm{A}}=0.1$ we have that {$\beta^*_{1-p_{\mathrm{A}}}(\sigma,\,\rho) \approx 0.44 < 1/2$} for a state $\rho$ with $\theta_0=\pi/3$. We thus see that it is not guaranteed that \emph{every} quantum classifier, which predicts $\sigma$ to be of class $k_{\mathrm{A}}$ with probability at least $0.9$, classifies $\rho$ to be of the same class. Now, we would like to find the maximum $\theta_0$, for which every classifier with confidence greater than $p_{\mathrm{A}}$ is guaranteed to classify $\rho$ and $\sigma$ equally. Using the fidelity bound~\eqref{eq:robustness_bound_fidelity}, we find the robustness condition on $\theta_0$
\begin{equation}
    \begin{aligned}
        \abs{\braket{\rho}{\sigma}}^{{2}} = \cos^2(\theta_0/2) > \frac{1}{2} + \sqrt{p_{\mathrm{A}}(1-p_{\mathrm{A}})}\\
        \iff \theta_0 < 2\cdot\arccos{\sqrt{\frac{1}{2} + \sqrt{p_{\mathrm{A}}(1-p_{\mathrm{A}})}}}.
    \end{aligned}
\end{equation}
In particular, if $p_{\mathrm{A}}=0.9$, we find that angles $\theta_0 < 2\cdot\arccos(\sqrt{0.8})\approx0.93 < \pi /3$ are certified.
Figure~\ref{fig:bloch_sphere} illustrates this scenario: the dark region around $\sigma$ contains all states $\rho$ for which is guaranteed that $\cA(\rho) = \cA(\sigma)$ for any classifier $\cA$ with confidence at least $0.9$.\\

{\it Classifier example.}
{We consider a binary quantum classifier $\cA$ which discriminates single-qubit states on the upper half of the Bloch sphere (class $+$) from states on the lower half (class $-$).} Specifically, we consider the dichotomic POVM $\{\Pi_{\theta,\phi},\,\Id_2 - \Pi_{\theta,\phi}\}$ defined by the projection operator $\Pi_{\theta,\phi} = \ketbra{\psi_{\theta,\phi}}{\psi_{\theta,\phi}}$ where
\begin{equation}
    \ket{\psi_{\theta,\phi}}:=\cos(\theta/2)\ket{0} + \sin(\theta/2)e^{i\phi}\ket{1}
\end{equation}
with $\theta=2\cdot\arccos(\sqrt{0.9})\approx 0.644$ and $\phi=\pi/2$. Furthermore, for the rest of this section, we assume that $p_{\mathrm{A}} + p_{\mathrm{B}} = 1$ so that $p_{\mathrm{B}}$ is determined by $p_{\mathrm{A}}$ via $p_{\mathrm{B}} = 1 - p_{\mathrm{A}}$. An illustration of this classification problem is given in Figure~\ref{fig:bloch_sphere}, where the decision boundary of $\cA$ is represented by the grey disk crossing the origin of the Bloch sphere. The states marked with a black $+$ correspond to $+$ states which have been classified correctly, states marked with a black $-$ sign correspond to data points correctly classified as $-$ and red states are misclassified by $\cA$. It can be seen that, since the state $\rho$ has been shown to violate the robustness condition (i.e. $\beta^*_{1-p_{\mathrm{A}}}(\sigma,\,\rho) \approx 0.44 < 1/2$), it is not guaranteed that $\rho$ and $\sigma$ are classified equally. In particular, for the example classifier $\cA$ we have $\cA(\rho) \neq \cA(\sigma)$.

In summary, as $p_{\mathrm{A}} \to \frac{1}{2}$, the robust radius approaches $0$. In the QHT view, this can be interpreted in the sense that if the type-I error probability $\alpha_0$ approaches $1/2$, then all alternative states can be discriminated from $\sigma$ with type-II error probability less than $1/2$. As $p_{\mathrm{A}} \to 1$, the robust radius approaches $\pi/2$. In this regime, the QHT view says that if the type-I error probability $\alpha_0$ approaches $0$, then the optimal type-II error probability is smaller than $1/2$ only for states in the lower half of the Bloch sphere.

\subsection{Robustness certification}
\label{subsec:robustness_certification}
The theoretical results in Section~\ref{subsec:explicit_bounds} provide conditions under which it is guaranteed that the output of a classification remains unaffected if the adversarial (noisy) state and the benign state are close enough, measured in terms of the fidelity, Bures metric, or trace distance.
Here, we show how this result can be put to work and make concrete examples of scenarios where reasoning about the robustness is relevant.
Specifically, we first present a protocol to assess how resilient a quantum classifier is against input perturbations.
Secondly, in a scenario where one is provided with a potentially noisy or adversarial input, we wish to obtain a statement as to whether the classification of the noisy input is guaranteed to be the same as the classification of a clean input without requiring access to the latter.
Thirdly, we analyse the robustness of quantum classifiers against known noise models, namely phase and amplitude damping.\\

{\it Assessing resilience against adversaries.}
In security critical applications, such as for example the classification of medical data or home surveillance systems, it is critical to assess the degree of resilience that machine learning systems exhibit against actions of malicious third parties. In other words, the goal is to estimate the expected classification accuracy, under perturbations of an input state within $1-\varepsilon$ fidelity. 
In the classical machine learning literature, this quantity is called the \emph{certified test set accuracy} at radius $r$, where distance is typically measured in terms of $\ell_{p}$-norms, and is defined as the fraction of samples in a test set which has been classified correctly and with a robust radius of at least $r$ (i.e. an adversary can not change the prediction with a perturbation of magnitude less than $r$).
We can adapt this notion to the quantum domain and, given a test set consisting of pairs of labelled samples $\cT=\{(\sigma_{\mathrm{i}},\,y_{\mathrm{i}})\}_{\mathrm{i}=1}^{\abs{\cT}}$, the \emph{certified test set accuracy} at fidelity $1-\varepsilon$ is given by
\begin{equation}
    \frac{1}{\abs{\cT}}\sum_{(\sigma,\,y)\in\cT}\Id\{\cA(\sigma)=y\,\land\,r_{\mathrm{F}}(\sigma) \leq 1-\varepsilon\}
\end{equation}
where $r_{\mathrm{F}}(\sigma)$ is the minimum robust fidelity~\eqref{eq:robustness_bound_fidelity} for sample $\sigma$ and $\Id$ denotes the indicator function.
To evaluate this quantity, we need to obtain the prediction and to calculate the minimum robust fidelity for each sample $\sigma\in\cT$ as a function of the class probabilities $\y_{k}(\sigma)$.
In practice, in the finite sampling regime, we have to estimate these quantities by sampling the quantum circuit $N$ times. To that end, we use Hoeffding's inequality so that the bounds hold with probability at least $1-\alpha$. Specifically, we run the following steps to certify the robustness for a given sample $\sigma$:
\begin{enumerate}
    \item Apply the quantum circuit $N$ times to $\sigma$ and perform the $\abs{\cC}$-outcome measurement $\{\Pi_{k}\}_{k=1}^{\abs{\cC}}$ each time. Store the outcomes in variables $n_{\mathrm{k}}$ for every $k\in\cC$.
    \item Determine the most frequent measurement outcome $k_{\mathrm{A}}$ and set $\hat{p}_{\mathrm{A}} = {n_{k_{\mathrm{A}}}}/{N} - \sqrt{-\ln(\alpha)/2N}$.
    \item If $\hat{p}_{\mathrm{A}}> 1/2$, set $\hat{p}_{\mathrm{B}} = 1 - \hat{p}_{\mathrm{A}}$ and calculate the minimum robust fidelity $r_{\mathrm{F}}$ according to~\eqref{eq:robustness_bound_fidelity} and return $(k_{\mathrm{A}},\,r_{\mathrm{F}})$; otherwise abstain from certification.
\end{enumerate}
Executing these steps for a given sample $\sigma$ returns the true minimum robust fidelity with probability $1-\alpha$, which follows from Hoeffding's inequality
\begin{equation}
    \mathrm{Pr}\left[\frac{n_{\mathrm{k}}}{N} - \langle\Lambda_{k}\rangle_\sigma \geq \delta\right] \leq \exp\{-2N\delta^2\}
\end{equation}
with $\Lambda_{k} = \cE^\dagger(\Pi_{k})$ and setting $\delta = \sqrt{-\ln(\alpha)/2N}$.
In supplementary note 6, this algorithm is shown in detail in Protocol 1.\\

{\it Certification for noisy inputs.}
In practice, inputs to quantum classifiers are typically noisy.
This noise can occur either due to imperfect implementation of the state preparation device, or due to an adversary which interferes with state or gate preparation.
Under the assumption that we know that the state has been prepared with fidelity at least $1-\varepsilon$ to the noiseless state, we would like to know whether this noise has altered our prediction, \emph{without having access to the noiseless state}.
Specifically, given the classification result, which is based on the \emph{noisy} input, we would like to have the guarantee that the classifier would have predicted the same class, had it been given the noiseless input state.
This would allow the conclusion that the result obtained from the noisy state has not been altered by the presence of noise.
To obtain this guarantee, we leverage Theorem~\ref{thm:fidelity_bound} in the following protocol. Let $\rho$ be a noisy input with $F(\rho,\,\sigma) > 1-\varepsilon$ where $\sigma$ is the noiseless state and let $\cA$ be a quantum classifier with quantum channel $\cE$ and POVM $\{\Pi_{k}\}_{k}$.
Similar to the previous protocol, we again need to take into account that in practice we can sample the quantum circuit only a finite number of times. Thus, we again use Hoeffding's inequality to obtain estimates for the class probability $p_{\mathrm{A}}$ which holds with probability at least $1-\alpha$. The protocol then consists of the following steps:
\begin{enumerate}
    \item Apply the quantum circuit $N$ times to the (noisy) state $\rho$ and perform the $\abs{\cC}$-outcome measurement $\{\Pi_{k}\}_{k=1}^{\abs{\cC}}$ each time. Store the outcomes in variables $n_{\mathrm{k}}$ for every $k\in\cC$.
    \item Determine the most frequent measurement outcome $k_{\mathrm{A}}$ and set $\hat{p}_{\mathrm{A}} = {n_{k_{\mathrm{A}}}}/{N} - \sqrt{-\ln(\alpha)/2N}$.
	\item If $\hat{p}_{\mathrm{A}}> 1/2$, set $\hat{p}_{\mathrm{B}} = 1 - \hat{p}_{\mathrm{A}}$ and calculate the minimum robust fidelity $r_{\mathrm{F}}$ according to~\eqref{eq:robustness_bound_fidelity} using $\hat{p}_{\mathrm{A}}$; otherwise, abstain from certification.
	\item If $1-\varepsilon > r_{\mathrm{F}}$, it is guaranteed that $\cA(\rho) = \cA(\sigma)$.
\end{enumerate}
Running these steps, along with a classification, allows to certify that the classification has not been affected by the noise, i.e. that the same classification outcome would have been obtained on the noiseless input state.\\

{\it Robustness for known noise models.}
Now, we analyse the robustness of a quantum classifier against known noise models which are parametrized by a noise parameter $\gamma$. Specifically, we investigate robustness against phase damping and amplitude damping.  Using Theorem~\ref{thm:fidelity_bound}, we calculate the fidelity between the clean input $\sigma$ and the noisy input $\cN_\gamma(\sigma)$ and rearrange the robustness condition~\eqref{eq:robustness_bound_fidelity} such that it yields a bound on the maximal noise which the classifier tolerates.

Phase damping describes the loss of quantum information without loosing energy. For example, it describes how electronic states in an atom are perturbed upon interacting with distant electrical charges. The quantum channel corresponding to this noise model can be expressed in terms of Kraus operators which are given by
\begin{equation}
    K_0 =
    \begin{pmatrix}
        1 & 0 \\
        0 & \sqrt{1-\gamma}
    \end{pmatrix},
    \ws
    K_1 =
    \begin{pmatrix}
        0 & 0\\
        0 & \sqrt{\gamma}
    \end{pmatrix}
\end{equation}
where $\gamma$ is the noise parameter. From this description alone, we can see that a system which is in the $\ket{0}$ or $\ket{1}$ state is always robust against all noise parameters in this model as it acts trivially on $\ket{0}$ and $\ket{1}$. Any such behaviour should hence be reflected in the tight robustness condition we derive from QHT. Indeed, for a pure state $\ket{\psi} = \alpha\ket{0} + \beta\ket{1}$, Theorem~\ref{thm:fidelity_bound} leads to the robustness condition $\gamma \leq 1$ if $\alpha=0$ or $\beta = 0$ and, for any $\alpha,\beta\neq 0$,
\begin{equation}
    \label{eq:phase_damping_bound}
    \gamma < 1 - \left(\max\left\{0,\,1 + \frac{{r_{\mathrm{F}} - 1}}{2\abs{\alpha}^2\abs{\beta}^2}\right\}\right)^2
\end{equation}
where $r_{\mathrm{F}} = \frac{1}{2}(1 + \sqrt{g(p_{\mathrm{A}},\,p_{\mathrm{B}})})$ is the fidelity bound from Theorem~\ref{thm:fidelity_bound} and $p_{\mathrm{A}},\,p_{\mathrm{B}}$ are the corresponding class probability bounds. 
This bound is illustrated in FIG.~\ref{fig:noise_model_bounds} as a function of $\abs{\alpha}^2$ and $p_{\mathrm{A}}$.
The expected behaviour towards the boundaries can be seen in the plot, namely that when $\abs{\alpha}^2 \to \{0,\,1\}$, then the classifier is robust under all noise parameters $\gamma \leq 1$.

Amplitude damping models effects due to the loss of energy from a quantum system (energy dissipation). For example, it can be used to model the dynamics of an atom which spontaneously emits a photon. The quantum channel corresponding to this noise model can be written in terms of Kraus operators
\begin{equation}
    K_0 =
    \begin{pmatrix}
        1 & 0 \\
        0 & \sqrt{1-\gamma}
    \end{pmatrix},
    \ws
    K_1 =
    \begin{pmatrix}
        0 & \sqrt{\gamma} \\
        0 & 0
    \end{pmatrix},
\end{equation}
where $\gamma$ is the noise parameter and can be interpreted as the probability of losing a photon.
It is clear from the Kraus decomposition that the $\ket{0}$ state remains unaffected. This again needs to be reflected by a tight robustness condition.
For a pure state $\ket{\psi} = \alpha\ket{0} + \beta\ket{1}$, Theorem~\ref{thm:fidelity_bound} leads to the robustness condition $\gamma \leq 1$ if $\abs{\alpha}=1$ and, for any $\alpha,\beta\neq 0$,
\begin{equation}
    \label{eq:amplitude_damping_bound}
    \begin{aligned}
        \gamma &< 1 - \Bigg[\frac{\abs{\alpha}^2}{\abs{\alpha}^2-\abs{\beta}^2}\cdot\Bigg( 1 - \\
        &\hspace{2em} \sqrt{1 - \frac{\abs{\alpha}^2 - \abs{\beta}^2}{\abs{\alpha}^2\abs{\beta}^2}\cdot\frac{\max\{0,{r_{\mathrm{F}}}-\abs{\alpha}^2\}}{\abs{\alpha}^2}}\,\Bigg)\Bigg]^2
    \end{aligned}
\end{equation}
where again $r_{\mathrm{F}} = \frac{1}{2}(1 + \sqrt{g(p_{\mathrm{A}},\,p_{\mathrm{B}})})$ is the fidelity bound from Theorem~\ref{thm:fidelity_bound}.
This bound is illustrated in FIG.~\ref{fig:noise_model_bounds} as a function of $\abs{\alpha}^2$ and $p_{\mathrm{A}}$. 
It can be seen again that the bound shows the expected behaviour, namely that when $\abs{\alpha}^2 \to 1$, then the classifier is robust under all noise parameters $\gamma \leq 1$.

We remark that, in contrast to the previous protocol, here we assume access to the noiseless state $\sigma$ and we compute the robustness condition on the noise parameter based on the classification of this noiseless state. This can be used in a scenario where a quantum classifier is developed and tested on one device, but deployed on a different device with different noise sources.

\begin{figure}[t]
    \centering
    \includegraphics[width=\linewidth]{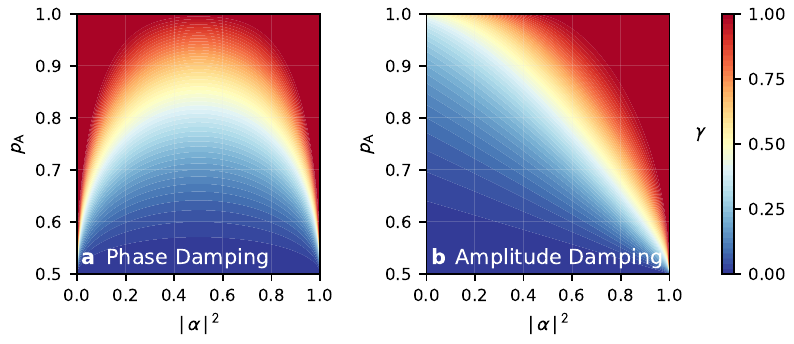}
    \caption{{Robustness against known noise models. Both plots show the maximal noise parameter $\gamma$ for which the classifier $\cA$ is still guaranteed to be robust, for \textbf{a} phase damping and \textbf{b} amplitude damping, when classifying a pure state input $\ket{\psi} = \alpha \ket{0} + \beta\ket{1}$. In \textbf{a}, we can see that for states $\ket{0}$ and $\ket{1}$, the classifier is robust against any $\gamma \leq 1$, while for \textbf{b} the same holds if the input state is $\ket{1}$.}}
    \label{fig:noise_model_bounds}
\end{figure}

\subsection{Randomized inputs with depolarization smoothing}
\label{sec:examples_noisy_input}
In the previous section, we looked at robustness of quantum classifiers against certain types of noise, either with respect to a known noise model, or with respect to unknown, potentially adversarial, noise. 
Here we take a different viewpoint, and investigate how robustness against unknown noise sources can be enhanced by harnessing depolarization noise.
This is led by the intuition that noise can be exploited to increase robustness and privacy.
We first provide background on randomized smoothing, a technique for provable robustness from classical machine learning. We then proceed to present provable robustness in terms of trace distance which is equivalent to the robustness condition~(\ref{eq:robustness_condition}) from Theorem~\ref{thm:main} but with {depolarized} inputs. The bound is then compared numerically with the H\"older duality bound from Lemma~\ref{lem: elementary bound} and with a result obtained recently from quantum differential privacy~\cite{du2020quantum}.\\

{\it Randomized smoothing.}
Randomized Smoothing is a technique that has recently been proposed to certify the robustness and obtain tight provable robustness guarantees in the classical setting~\cite{cohen2019certified}. The key idea is to randomize inputs to classifiers by perturbing them with additive Gaussian noise. This results in smoother decision boundaries which in turn leads to improved robustness to adversarial attacks. In this section, we extend this concept to the quantum setting by interpreting quantum noise channels as ``smoothing'' channels. The idea of harnessing actively induced input noise in quantum classifiers to increase robustness has recently been proposed in Ref.~\cite{du2020quantum} where a robustness bound with techniques from quantum differential privacy has been derived.
In the following, we take a similar path and consider a depolarization noise channel and analytically derive a larger robustness radius for pure single-qubit input states.\\

{\it Quantum channel smoothing: depolarization.}
Consider depolarization noise which maps a state $\sigma$ onto a linear combination of itself and the maximally mixed state
\begin{equation}
    \sigma \mapsto \cE^{\mathrm{dep}}_p(\sigma) := (1-p)\sigma + \frac{p}{d}\Id_{d}
\end{equation}
where $p\in(0,\,1)$ is the depolarization parameter and $d$ is the dimensionality of the underlying Hilbert space. In single-qubit scenarios, this can geometrically be interpreted as a uniform contraction of the Bloch sphere parametrized by $p$, pushing quantum states towards the completely mixed state. Analogously to classical randomized smoothing, we apply a depolarization channel to inputs before passing them through the classifier in order to artificially randomize the states and increase robustness against adversarial attacks. We then obtain a robustness guarantee by instantiating Theorem~\ref{thm:main} in the following way. {Let $\sigma$ be a benign input state and suppose that the classifier $\cA$ with score function $\y$ satisfies}
\begin{equation}
    \y_{k_{\mathrm{A}}}(\cE^{\mathrm{dep}}_p(\sigma)) \geq p_{\mathrm{A}} > p_{\mathrm{B}} \geq \max_{k \neq k_{\mathrm{A}}}\y_{k}(\cE^{\mathrm{dep}}_p(\sigma)).
\end{equation}
{Then $\cA$ is robust at $\cE^{\mathrm{dep}}_p(\rho)$ for any adversarial input state $\rho$ which satisfies the robustness condition~\eqref{eq:robustness_condition}, where $\beta^*$ is the optimal type-II error probability for testing $\cE^{\mathrm{dep}}_p(\sigma)$ against $\cE^{\mathrm{dep}}_p(\rho)$.} 
In particular, if $\sigma$ and $\rho$ are single-qubit pure states and in the case where we have $p_{\mathrm{A}} + p_{\mathrm{B}} = 1$, the robustness condition can be equivalently expressed in terms of the trace distance as $T(\rho,\,\sigma) < {r_{\mathrm{Q}}(p)}$ with
\begin{equation}
    \label{eq:depolarization_bound_main}
    r_{\mathrm{Q}}(p) =
        \begin{cases}
            \sqrt{\frac{1}{2}-\frac{\sqrt{{g(p,\,p_{\mathrm{A}})}}}{1-p}}, &\ p_{\mathrm{A}} < \frac{1 + 3(1-p)^2}{2 + 2(1-p)^2}\\
            \sqrt{\frac{p\cdot(2-p)\cdot(1-2p_{\mathrm{A}})^2}{8(1-p)^2\cdot(1-p_{\mathrm{A}})}}, &\ p_{\mathrm{A}} \geq \frac{1 + 3\cdot(1-p)^2}{2+2\cdot(1-p)^2}
        \end{cases}
\end{equation}
where
\begin{equation}
    {g(p,\,p_{\mathrm{A}})} = \frac{1}{2}\left(2p_{\mathrm{A}}(1-p_{\mathrm{A}}) - p(1-\frac{p}{2})\right).
\end{equation}
A detailed derivation of this bound is given in supplementary note 5.

The H\"older bound from Lemma~\ref{lem: elementary bound} can also be adapted to the noisy setting. Specifically, since for two states $\sigma$ and $\rho$, the trace distance obeys $T(\cE^{\mathrm{dep}}_p(\rho),\,\cE^{\mathrm{dep}}_p(\sigma)) = (1-p)\cdot T(\rho,\,\sigma)$, Lemma~\ref{lem: elementary bound} implies robustness given that the trace distance is less than $T(\rho,\,\sigma) < {r_{\mathrm{H}}(p)}$ where
\begin{equation}
    \label{eq:holder_bound_noisy}
    r_{\mathrm{H}}(p) = \frac{2p_{\mathrm{A}}-1}{2(1-p)}.
\end{equation}

It has been shown in Ref.~\cite{du2020quantum} that naturally occurring noise in a quantum circuit can be harnessed to increase the robustness of quantum classification algorithms. Specifically, using techniques from quantum differential privacy, a robustness bound expressible in terms of the class probabilities $p_{\mathrm{A}}$ and the depolarization parameter $p$ has been derived. Written in our notation and for single-qubit binary classification, the bound can be written as
\begin{equation}
    r_{\mathrm{DP}}(p) = \frac{p}{2(1-p)}\left(\sqrt{\frac{p_{\mathrm{A}}}{1-p_{\mathrm{A}}}} - 1\right)
\end{equation}
and robustness is guaranteed for any adversarial state $\rho$ with $T(\rho,\,\sigma) < {r_{\mathrm{DP}}(p)}$.
The three bounds are compared graphically in FIG.~\ref{fig:depolarization_bounds} for different values of the noise parameter $p$, showing that the QHT bound gives rise to a tighter robustness condition for all values of $p$.

\begin{figure}[t]
    \centering
    \includegraphics[width=\linewidth]{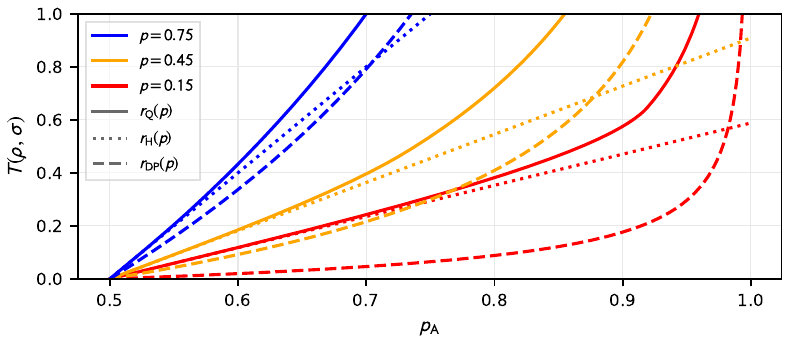}
    \caption{Comparison of Robustness bounds for single-qubit pure states derived from quantum hypothesis testing $r_{\mathrm{Q}}(p)$ , H\"older duality $r_{\mathrm{H}}(p)$ and quantum differential privacy $r_{\mathrm{DP}}(p)$~\cite{du2020quantum} with different levels of depolarization noise $p$.}
    \label{fig:depolarization_bounds}
\end{figure}

It is worth remarking that although the QHT robustness bounds can be, as shown here for the case of applying depolarization channel, enhanced by {active} input randomization, it already presents a valid, non-trivial condition with {noiseless} (without smoothing) quantum input ({Theorems~\ref{thm:main},~\ref{thm:fidelity_bound}, Corollary~\ref{cor:trace_bound_pure_mixed} and Lemma~\ref{lem: elementary bound}}). This contrasts with the {deterministic} classical scenario, where the addition of classical noise sources to the input state is necessary to generate a probability distribution corresponding to the input data, from which an adversarial robustness bound can be derived \cite{cohen2019certified}. This distinction between the quantum and classical settings roots in the probabilistic nature of measurements on quantum states, which of course applies to both pure and mixed state inputs.

\section{Discussion}
\label{sec:discussion}

We have seen how a fundamental connection between adversarial robustness of quantum classifiers and quantum hypothesis testing (QHT) can be leveraged to provide a powerful framework for deriving optimal conditions for robustness certification.
The robustness condition is provably tight when expressed in the SDP formulation in terms of optimal error probabilities for binary classifications or, more generally, for multiclass classifications where the probability of the most likely class is greater than $1/2$.
The corresponding closed form expressions arising from the SDP formulation are proved to be tight for general states when expressed in terms of fidelity and Bures distance, whereas in terms of trace distance, tightness holds only for pure states.
These bounds give rise to (1) a practical robustness protocol for assessing the resilience of a quantum classifier against adversarial and unknown noise sources; (2) a protocol to verify whether a classification given a noisy input has had the same outcome as a classification given the noiseless input state, without requiring access to the latter, and (3) conditions on noise parameters for amplitude and phase damping channels, under which the outcome of a classification is guaranteed to remain unaffected.
Furthermore, we have shown how using a randomized input with depolarization channel enhances the QHT bound, consistent with previous results, in a manner akin to randomized smoothing in robustness certification of classical machine learning. 

A key difference between the quantum and classical formalism is that quantum states themselves have a naturally probabilistic interpretation, even though the classical data that could be embedded in quantum states do not need to be probabilistic. 
We now know that both classical and quantum optimal robustness bounds for classification protocols depend on bounds provided by hypothesis testing. However, hypothesis testing involves the comparison of probability distributions, which can only be possible in the classical case with the addition of stochastic noise sources if the classical data is initially non-stochastic. This means that the optimal robustness bounds in the classical case only exist for noisy classifiers which also require training under the additional noise \cite{cohen2019certified}. This is in contrast to the quantum scenario. Our quantum adversarial robustness bound can be proved independently of randomized input, even though it can be enhanced by it, like through a depolarization channel. Thus, in the quantum regime, unlike in the classical {deterministic} scenario, we are not forced to consider training under actively induced noise.

Our optimal provable robustness bound and the connection to quantum hypothesis testing also provides a first step towards more rigorously identifying the limitations of quantum classifiers in its power of distinguishing between quantum states. Our formalism hints at an intimate relationship between these fundamental limitations in the accuracy of distinguishing between different classes of states and robustness. This could shed light on the robustness and accuracy trade-offs observed in classification protocols  \cite{tsipras2019robustness} and is an important direction of future research. It is also of independent interest to explore possible connections between tasks that use quantum hypothesis testing, such as quantum illumination \cite{wilde2017} and state discrimination~\cite{sentis2019}, with accuracy and robustness in quantum classification. 

\section{Methods}
\subsection{Proof of Theorem~\ref{thm:main}}
    The proof of this theorem is based on showing that the measurement operators of the classifier can be viewed as an operator which is feasible for the SDP~\eqref{eq:type2_error_sdp}. Specifically, note that in the Heisenberg picture we can write the score function $\y$ of the classifier $\cA$ as
    \begin{equation}
        \y_{k}(\sigma) = \Tr{\cE^\dagger\left(\Pi_{k}\right)\sigma} = \Tr{\Lambda_{k}\sigma}
    \end{equation}
    where $\Lambda_{k}:=\cE^\dagger(\Pi_{k})$. Since $\cE$ is a CPTP map, its dual is completely positive and unital and thus $0 \leq \Lambda_{k} \leq \Id$ and
    \begin{equation}
        \sum_k \Lambda_{k} = \sum_k \cE^\dagger(\Pi_{k}) = \cE^\dagger(\Id) = \Id.
    \end{equation}
    Note that the operator $\Id - \Lambda_{k_{\mathrm{A}}}$ is feasible for the SDP $\beta^*_{1-p_{\mathrm{A}}}(\sigma,\,\rho)$ since by assumption
    \begin{equation}
        \alpha(\Id - \Lambda_{k_{\mathrm{A}}};\,\sigma) = 1 - \y_{k_{\mathrm{A}}}(\sigma) \leq 1-p_{\mathrm{A}}.
    \end{equation}
    It follows that
    \begin{equation}
        \y_{k_{\mathrm{A}}}(\rho) = \beta(\Id - \Lambda_{k_{\mathrm{A}}};\,\rho) \geq \beta^*_{1-p_{\mathrm{A}}}(\sigma,\,\rho).
    \end{equation}
    Similarly, let $k\neq k_{\mathrm{A}}$ be arbitrary. Then, the operator $\Lambda_{k}$ is feasible for the SDP $\beta^*_{p_{\mathrm{B}}}(\sigma,\,\rho)$ since
    \begin{equation}
        \alpha(\Lambda_{k};\,\sigma) = \y_{k}(\sigma) \leq p_{\mathrm{B}}
    \end{equation}
    and hence
    \begin{equation}
        1 - \y_{k}(\rho) = \beta(\Lambda_{k};\,\rho) \geq \beta^*{p_{\mathrm{B}}}(\sigma,\,\rho)
    \end{equation}
    Since $k\neq k_{\mathrm{A}}$ is arbitrary, it follows that if $\rho$ satisfies
    \begin{equation}
        \beta^*_{1-p_{\mathrm{A}}}(\sigma,\,\rho) + \beta^*{p_{\mathrm{B}}}(\sigma,\,\rho) > 1
    \end{equation}
    then it is guaranteed that
    \begin{equation}
        \y_{k_{\mathrm{A}}}(\rho) > \max_{k \neq k_{\mathrm{A}}} \y_{k}(\rho)
    \end{equation}
    and thus $\cA(\rho) = \cA(\sigma)$.
    \null\nobreak\hfill\ensuremath{\square}

\subsection{Proof of Theorem~\ref{thm:tightness}}
    Note that, since $p_{\mathrm{B}} = 1-p_{\mathrm{A}}$ by assumption, the robustness condition~\eqref{eq:robustness_condition} reads
    \begin{equation}
        \label{eq:tightness_proof_rob_condition}
        \beta^*_{1-p_{\mathrm{A}}}(\sigma,\,\rho) > 1/2.
    \end{equation}
    Let $M_{\mathrm{A}}^\star$ be an optimizer of the corresponding SDP such that $\alpha(M_{\mathrm{A}}^\star) = 1-p_{\mathrm{A}}$ and
    \begin{equation}
        \beta(M_{\mathrm{A}}^\star;\,\rho) = \beta^*_{1-p_{\mathrm{A}}}(\sigma,\,\rho).
    \end{equation}
    Consider the classifier $\cA^\star$ with score function $\y^\star$ defined by the POVM $\{\Id - M_{\mathrm{A}}^\star,\,M_{\mathrm{A}}^\star,\,0\}$ where the number of $0$ operators is such that $\y$ has the desired number of classes. The score function $\y^\star$ is consistent with the class probabilities~\eqref{eq:class_probs} since
    \begin{align}
        \y^\star_{k_{\mathrm{A}}}(\sigma) &= \alpha(\Id - M_{\mathrm{A}}^\star;\,\sigma) = p_{\mathrm{A}}\\
        \y^\star_{k_\mathrm{B}}(\sigma) &= \alpha(M_{\mathrm{A}}^\star;\,\sigma) = 1- p_{\mathrm{A}} = p_{\mathrm{B}}.
    \end{align}
    Furthermore, if $\rho$ violates~\eqref{eq:tightness_proof_rob_condition}, then we have
    \begin{equation}
        \y_{k_{\mathrm{A}}}(\rho) = \beta(M_{\mathrm{A}}^\star;\,\rho) \leq 1/2
    \end{equation}
    and thus, in particular $\cA^\star(\rho)\neq k_{\mathrm{A}} = \cA^\star(\sigma)$.
    \null\nobreak\hfill\ensuremath{\square}
    
\subsection{Fidelity robustness condition}
Recall that the robustness condition in Theorem~\ref{thm:main} is expressed in terms of the SDP from the Neyman-Pearson approach to quantum hypothesis testing.
Thus, in order to use Theorem~\ref{thm:main} to obtain robustness bounds in terms of a meaningful distance between quantum states, we need to connect the optimal type-II error with this distance. 
Here, we look specifically at the {fidelity} between pure quantum states and sketch the proof for Lemma~\ref{lem:pure_state_bound}.
We refer the reader to supplementary note 3 for details. 

\begin{proof}[Proof of Lemma~\ref{lem:pure_state_bound} (sketch)]
    The key challenge to proving this result is connecting the robustness condition~\eqref{eq:robustness_condition}, written in terms of type-II error probabilities, to the {fidelity $F$ which, for pure states, is given by the squared overlap $\abs{\braket{\psi_\sigma}{\psi_\rho}}^2$.
    It is well known that optimizers to the SDP~\eqref{eq:type2_error_sdp} are given by Helstrom operators, $M_t$
    which can be expressed in terms of the projection onto the positive and null eigenspaces of the operator $\rho - t\sigma$.}
    The first step is thus to solve the eigenvalue problem
    \begin{equation}
        (\rho - t\sigma)\ket{\eta} = \eta\ket{\eta}
    \end{equation}
    which, for pure states, can be expressed in terms of the squared overlap $\abs{\braket{\psi_\sigma}{\psi_\rho}}^2$.
    Given these solutions, one then derives an expression for the Helstrom operators $M^\star_{\mathrm{A}}$ and $M^\star_{\mathrm{B}}$ with type-I error probabilities $1-p_{\mathrm{A}}$ and $p_{\mathrm{B}}$ respectively. This leads to the robustness condition
    $
        \beta(M^\star_{\mathrm{A}};\,\rho) + \beta(M^\star_{\mathrm{B}};\,\rho) > 1
    $
    being an inequality which can be rewritten as a condition on the fidelity which takes the desired form~\eqref{eq:fidelity_bound_pure}.
\end{proof}
In a similar manner, one can derive the trace-distance bound for depolarized input states presented in the `Results' section of this paper. The full proof for the robustness bound in equation~\eqref{eq:depolarization_bound_main} is given in supplementary note 5.

\section*{ACKNOWLEDGEMENTS}
The authors are grateful to Ryan LaRose (Michigan State University), Zi-Wen Liu (Perimeter Institute for Theoretical Physics), Barry Sanders (University of Calgary) and Robert Pisarczyk (University of Oxford)
for inspiring discussions on the question of robustness in quantum machine learning.

NL acknowledges funding from the Shanghai Pujiang Talent Grant (no. 20PJ1408400) and the NSFC International Young Scientists Project (no. 12050410230). NL is also supported by the Innovation Program of the Shanghai Municipal Education Commission (no. 2021-01-07-00-02-E00087) and the Shanghai Municipal Science and Technology Major Project (2021SHZDZX0102).

\section*{AUTHOR CONTRIBUTIONS}
The main idea was conceived by C.Z. in discussion with Z.Z. and M.W. 
Key insights to adversarial quantum learning were provided by N.L. while B.L. contributed to central insights to robustness in machine learning.
The work on QHT, the resulting QHT condition and the derivation for closed form bounds for pure and depolarized states was completed by M.W.
The extension of the fidelity bound the the mixed state case was completed by M.W. and Z.Z. The different trace distance bounds for mixed states were derived by Z.Z. The proof for optimality was done by M.W. and Z.Z. The noisy input scenario and the example were initiated by N.L. and completed by M.W.
All authors contributed to the manuscript.

\section*{COMPETING INTERESTS}
The authors declare no competing interests.

\bibliographystyle{naturemag}
\bibliography{main}

\clearpage
\newpage
\onecolumngrid
\renewcommand{\thesection}{Supplementary Note \arabic{section}}
\renewcommand\thesubsection{\thesection.\Alph{subsection}}
\def\bibsection{\section*{Supplementary References}}

\begin{center}
    \large\bf
    Supplementary Information for\\
    ``Optimal Provable Robustness of Quantum Classification via Quantum Hypothesis Testing''
\end{center}
\vspace{2em}

Here, we provide detailed proofs for the robustness bounds presented in the paper in terms of the fidelity and trace distance, stated in Lemma~\ref{lem:pure_state_bound}, Corollary~\ref{cor:trace_bound_pure_mixed} and the bound for depolarized inputs (Eq.~\eqref{eq:depolarization_bound_main} in the main part).
To that end, we first show a collection of technical lemmas related to quantum hypothesis testing with the goal of explicitly constructing Helstrom operators which attain a preassigned level of type-I error probability. 
These constructions of the Helstrom operators will then be used to derive an expression for the SDP (Eq.~\eqref{eq:type2_error_sdp} in the main part) in terms of the fidelity between the input states.

\section{Technical Lemmas}
\label{sec:supplemental_technical_lemmas}
{\it Preliminaries.} We first recall the central quantities of interest. As in the main part of this paper, the null hypothesis corresponds to a benign input state and is described by a density operator $\sigma\in\cS(\cH)$ acting on a Hilbert space $\cH$ of finite dimension $d:=\mathrm{dim}(\cH) <\infty$. 
The density operator for the alternative hypothesis is denoted by $\rho$ and corresponds the adversarial state in the classification setting. A quantum hypothesis test is defined by a positive semi-definite operator $0\leq M\leq \Id_d$ and the type-I and type-II error probabilities associated with $M$ are denoted by $\alpha$ and $\beta$ and are defined by
\begin{align}
    \alpha(M;\,\sigma) &:= \Tr{\sigma M}\ws&\text{(type-I error)}\\
    \beta(M;\,\rho) &:= \Tr{\rho(\Id - M)} &\text{(type-II error)}
\end{align}
Throughout this supplementary information we will omit the explicit dependence on $\sigma$ and $\rho$ whenever it is clear from context.
For two Hermitian operators $A$ and $B$, we write $A\geq B$ ($A\leq B$) if $A-B$ is positive (negative) semi-definite and $A > B$ ($A < B$) if $A-B$ is positive (negative) definite.
For a Hermitian operator $A$ with spectral decomposition $A = \sum_i \lambda_i P_i$ with eigenvalues $\{\lambda_i\}_i$ and orthogonal projections onto the associated eigenspaces $\{P_i\}_i$, we write
\begin{equation}
    \{A > 0\}:=\sum_{i\colon\lambda_i > 0}P_i,\ws \{A < 0\}:=\sum_{i\colon\lambda_i < 0}P_i
\end{equation}
for the projections onto the eigenspaces associated with positive and negative eigenvalues respectively. Finally, for $t\geq0$ define the operators
\begin{equation}
    P_{t,+} := \{\rho-t\sigma > 0\},\ws P_{t,-} := \{\rho-t\sigma < 0\},\ws P_{t,0}:=\Id - P_{t,+} - P_{t,-}.
\end{equation}
Helstrom operators are then defined as
\begin{equation}
    M_t := P_{t,+} + X_t,\ws 0 \leq X_t \leq P_{t,0}.
\end{equation}
Finally, recall the SDP $\beta^*$ (Eq.~\eqref{eq:type2_error_sdp} in the main part) in the Neyman-Pearson approach to quantum hypothesis testing which, for $\alpha_0\in[0,\,1]$, is defined as
\begin{equation}
	\label{eq:type2_error_sdp_apx}
    \begin{aligned}
        \beta^*_{\alpha_0}(\sigma,\,\rho) := \mathrm{minimize}\ws&\beta(M;\,\rho)\\
        \mathrm{s.t.}\ws&\alpha(M;\,\sigma) \leq \alpha_0,\\
        \ws & 0 \leq M \leq \Id_d.
    \end{aligned}
\end{equation}
The following Lemmas lead to an explicit construction of Helstrom operators attaining a preassigned level of type-I error probability $\alpha_0$ and which are optimizers of the SDP $\beta^*$ (Eq.~\eqref{eq:type2_error_sdp} in the main part).
We first show that if a sequence of bounded Hermitian operators $A_n$ converges in operator norm to a bounded Hermitian operator $A$ from above (below), then the projection $\{A_n < 0\}$ ($\{A_n > 0\}$) converges to $\{A < 0\}$ ($\{A > 0\}$) in operator norm. This subsequently allows us to show that the function $t\mapsto \alpha(P_{t,+})$ is non-increasing and continuous from the right, and that $t\mapsto \alpha(P_{t,+} + P_{t,0})$ is non-increasing and continuous from the left. As a consequence, for $\alpha_0 \in [0,\,1]$, the quantity
\begin{equation}
    \tau(\alpha_0):= \inf\{t\geq 0\colon\,\alpha(P_{t,+}) \leq \alpha_0\}
\end{equation}
is well defined. This implies the chain of inequalities
\begin{equation}
    \alpha\left(P_{\tau(\alpha_0),+}\right) \leq \alpha_0 \leq \alpha\left(P_{\tau(\alpha_0),+} + P_{\tau(\alpha_0),0}\right).
\end{equation}
Based on this, we construct a Helstrom operator $M_{\tau(\alpha_0)}$ according to
\begin{equation}
    M_{\tau(\alpha_0)}:=P_{\tau(\alpha_0),+} + q_0P_{\tau(\alpha_0),0},
    \ws 
    q_0 := 
        \begin{cases}
            \frac{\alpha_0 - \alpha(P_{\tau(\alpha_0),+})}{\alpha(P_{\tau(\alpha_0),0})}\,&\mathrm{if}\ \alpha\left(P_{\tau(\alpha_0),0}\right) \neq 0,\\
            0 &\mathrm{otherwise.}
        \end{cases}
\end{equation}
which attains the preassigned type-I error probability $\alpha_0$. We will show that these Helstrom operators are optimal for the SDP $\beta^*_{\alpha_0}(\sigma,\,\rho)$ in~\eqref{eq:type2_error_sdp_apx}, so that
\begin{equation}
    \label{eq:helstrom_optimality}
    \beta^*_{\alpha_0}(\sigma,\,\rho) = \beta(M_{\tau(\alpha_0)};\,\rho).
\end{equation}

\begin{lem}
    \label{lem:operator_trace_convergence}
    Denote by $\cB(\cH)$ the space of bounded linear operators acting on the finite dimensional Hilbert space $\cH$, $d:=\mathrm{dim}(\cH)<\infty$.
    Let $A\in\cB(\cH)$ and $\{A_n\}_{n\in\bN}\subset\cB(\cH)$ be Hermitian operators and suppose that $\norm[op]{A_n-A} \xrightarrow{n\to\infty}0$. Then, it holds that
    \begin{align}
        (i)&\ A - A_n \leq 0 \, \Rightarrow \, \norm[op]{\{A_n < 0\} - \{A < 0\}} \xrightarrow{n\to\infty}0,\\
        (ii)&\ A - A_n \geq 0 \, \Rightarrow \, \norm[op]{\{A_n > 0\} - \{A > 0\}} \xrightarrow{n\to\infty}0.
    \end{align}
\end{lem}
\begin{proof}
    We first show that convergence in operator norm implies that the eigenvalues of $A_n$ converge towards the eigenvalues $A$. 
    For a linear operator $M$ let $\lambda_k(M)$ denote its $k$-th largest eigenvalue, $\lambda_1(M) \geq \ldots \geq \lambda_q(M)$, where $q \leq d$ is the number of distinct eigenvalues of $M$.
    By the minimax principle (e.g.~\cite{bhatia1997}, chapter 3),
    we can compute $\lambda_k$ for any Hermitian operator $M$ according to
    \begin{equation}
        \label{eq:minmax_eigenvalue}
        \lambda_k(M) = \max_{\substack{V\subseteq\cH\\\mathrm{dim}(V) = k}}\min_{\substack{\psi\in V\\\norm{\psi} = 1}}\bra{\psi}M\ket{\psi}.
    \end{equation}
    Now let $\varepsilon>0$ and let $n\in\bN$ large enough such that $\norm[op]{A_n - A} < \varepsilon$. Let $\ket{\psi}\in\cH$ be a normalized state and note that by the Cauchy-Schwartz inequality we have
    \begin{align}
        \abs{\bra{\psi}(A_n - A)\ket{\psi}} \leq \norm{(A_n - A)\psi}\norm{\psi} \leq \norm[op]{A_n-A}\norm{\psi}^2 < \varepsilon
    \end{align}
    and thus
    \begin{align}
        \bra{\psi}A\ket{\psi} - \varepsilon < \bra{\psi}A_n\ket{\psi} < \bra{\psi}A\ket{\psi} + \varepsilon.
    \end{align}
    Hence, for any fixed $k \geq 1$ and any subspace $V\subset \cH$ with $\mathrm{dim}(V) = k$, we have
    \begin{align}
        \min_{\substack{\psi\in V\\\norm{\psi} = 1}}\bra{\psi}A\ket{\psi} - \varepsilon < \min_{\substack{\psi\in V\\\norm{\psi} = 1}}\bra{\psi}A_n\ket{\psi} < \min_{\substack{\psi\in V\\\norm{\psi} = 1}}\bra{\psi}A\ket{\psi} + \varepsilon
    \end{align}
    and thus, from~(\ref{eq:minmax_eigenvalue}), we see that
    \begin{equation}
        \lambda_k(A) - \varepsilon < \lambda_k(A_n) < \lambda_k(A) - \varepsilon \Rightarrow \abs{\lambda_k(A) - \lambda_k(A_n)} < \varepsilon
    \end{equation}
    and hence 
    \begin{equation}
        \lambda_k(A_n) \xrightarrow{n\to\infty} \lambda_k(A),\hspace{2em}k=1,\,\ldots,\,q.
    \end{equation}
    Alternatively, this can be seen from Weyl's Perturbation Theorem (e.g.~\cite{bhatia1997}, ch. 3): namely, since $A$ and $A_n$ are Hermitian, it follows immediately from $\abs{\lambda_k(A) - \lambda_k(A_n)} \leq \max_k\abs{\lambda_k(A) - \lambda_k(A_n)} \leq \norm[op]{A -A_n}$ that eigenvalues converge provided that $\norm[op]{A_n - A} \to 0$. We will now make use of function theory and the resolvent formalism to show the convergence of the positive and negative eigenprojections. Let $M\in\cB(\cH)$ be Hermitian, let $\sigma(M)$ denote the spectrum of $M$ and, for $\lambda\in\bC\setminus\sigma(M)$, let
    \begin{equation}
        R_\lambda(M) := (M - \lambda\Id)^{-1} = -\sum_{k=0}^\infty \lambda^{-(k+1)}M^k
    \end{equation}
    be the resolvent of the operator $M$. The sum is the Neumann series and converges for $\lambda\in\bC\setminus\sigma(M)$. Since $M$ is Hermitian, we can write its spectral decomposition in terms of contour integrals over the resolvent
    \begin{equation}
        \begin{gathered}
            M = \sum_{k=1}^q \lambda_k(M) P_k\hspace{1em} \text{with} \hspace{1em} P_k = \frac{1}{2\pi \ci} \oint_{(\gamma_k,-)} R_\lambda(M) \,d\lambda,
            \hspace{1em} \text{and} \hspace{1em}\sum_{k=1}^q P_k = \Id
    \end{gathered}
    \end{equation}
    where $P_k$ is the orthogonal projection onto the $k$-th eigenspace and
    the integration is to be understood element-wise.
    The symbol $(\gamma_k,-)$ indicates that the contour encircles $\lambda_k(M)$ once negatively, but does not encircle any other eigenvalue of $M$.
    We refer the reader to~\cite{richtmeyer78} for a detailed derivation.\\
    
    We now show part $(i)$ of the Lemma.
    For ease of notation, let $\lambda_k$ denote the $k$-th eigenvalue of $A$ and $\lambda_{k,n}$ the $k$-th eigenvalue of $A_n$.
    Since $A_n$ and $A$ are Hermitian operators, we can write the eigenprojections $\{A_n < 0\}$ and $\{A < 0\}$ in terms of the resolvent as
    \begin{equation}
        \begin{gathered}
                \{A < 0\} = \frac{1}{2\pi \ci}\sum_{k\colon\lambda_k < 0} \oint_{(\gamma_k,-)} R_\lambda(A) \,d\lambda,\hspace{1em}
                \{A_n < 0\} = \frac{1}{2\pi \ci}\sum_{k\colon\lambda_{k,n} < 0} \oint_{(\gamma_{k,n},-)} R_\lambda(A_n) \,d\lambda
        \end{gathered}
    \end{equation}
    where the symbols $(\gamma_k,-)$ and $(\gamma_{k,n},-)$ indicate that the contours encircle only $\lambda_k$ and $\lambda_{k,n}$ once negatively and no other eigenvalues of $A$ and $A_n$ respectively.
    Since by assumption $A_n \geq A$ and $A_n,\,A$ are Hermitian, it follows from Weyl's Monotonicity Theorem that $\lambda_{k,n} \geq \lambda_k$.
    Let $\lambda_K$ be the largest negative eigenvalue of $A$, that is
    $\lambda_1 \geq \lambda_2 \geq \ldots \lambda_{K-1} \geq 0 > \lambda_K\geq \ldots \geq \lambda_q$.
    Note that if $A$ is positive semidefinite, then so is $A_n$ and the statement follows trivially from $\{A_n < 0\} = \{A < 0\} = 0$.
    Thus, without loss of generality, we can assume that at least one eigenvalue of $A$ is negative.
    Since $\lambda_{k,n} \geq \lambda_k$, and in particular $\lambda_{K-1,n} \geq \lambda_{K-1} \geq 0$, 
    there exists $N_0\in\bN$ large enough such that $\lambda_{K-1,n} \geq 0 > \lambda_{K,n}$ for all $n\geq N_0$.
    Let $r_0$ be the smallest distance between two eigenvalues of $A$
    \begin{align}
        r_0 := \min_k\abs{\lambda_k - \lambda_{k+1}}
    \end{align}
    and let $0 < \varepsilon < \frac{r_0}{2}$.
    Choose $N_1 \geq N_0$ large enough such that $\max_{k \geq K} \abs{\lambda_{k,n} - \lambda_k} < \varepsilon / 2$. 
    Let $0 < \delta < \frac{r_0}{2} - \varepsilon$ and for $k\geq K$ let $B_{\delta + \varepsilon}^k := B_{\delta + \varepsilon}(\lambda_k)$ be the open ball of radius $\delta + \varepsilon$ centered at $\lambda_k$. Note that $\partial B_{\delta + \varepsilon}^k$ encircles both $\lambda_{k,n}$ and $\lambda_{k}$.
    Then, for $k\geq K$ and $n\geq N_1$, the mappings
    \begin{align}
        &\lambda \mapsto R_{\lambda}(A),\hspace{1em} \lambda\in B_{\delta + \varepsilon}^k\setminus\{\lambda_k\},\\
        &\lambda \mapsto R_{\lambda}(A_n),\hspace{1em} \lambda\in B_{\delta + \varepsilon}^k\setminus\{\lambda_{k,n}\}
    \end{align}
    are holomorphic functions of $\lambda$ and each has an isolated (simple) singularity at $\lambda_k$ and $\lambda_{k,n}$ respectively.
    Let $\gamma_{k,n}$ be a contour around $\lambda_{k,n}$ encircling no other eigenvalue of $A_n$.
    Note that the contours $\gamma_{k,n}$ and $\partial B_{\delta + \varepsilon}^k$ are homotopic (in $B_{\delta + \varepsilon}^k\setminus\{\lambda_{k,n}\}$).
    Thus, for all $k\geq K$ and $n\geq N_1$, Cauchy's integral Theorem yields
    \begin{align}
        \oint_{\gamma_{k,n}} R_\lambda(A_n)\,d\lambda = \oint_{\partial B_{\delta + \varepsilon}^k} R_\lambda(A_n)\,d\lambda.
    \end{align}
    With this, we see that for $n\geq N_1$
    \begin{align}
        \{A_n < 0\} - \{A < 0\} &= \frac{1}{2\pi \ci}\sum_{k=K}^q\left(\oint_{\partial B_{\delta + \varepsilon}^k}R_\lambda(A)\,d\lambda - \oint_{\gamma_{k,n}}R_\lambda(A_n)\,d\lambda\right)\\
        &=\frac{1}{2\pi \ci}\sum_{k=K}^q\left(\oint_{\partial B_{\delta + \varepsilon}^k}\left(R_\lambda(A)-R_\lambda(A_n)\right)\,d\lambda\right)
    \end{align}
    and thus, by the triangle inequality
    \begin{align}
        \norm[op]{\{A_n < 0\} - \{A < 0\}} &\leq \frac{1}{2\pi}\sum_{k=K}^q\norm[op]{\oint_{\partial B_{\delta + \varepsilon}^k}\left(R_\lambda(A)-R_\lambda(A_n)\right)\,d\lambda}\\
        &\hspace{-2em}\leq \frac{1}{2\pi}\sum_{k=K}^q \sup_{\lambda\in\partial B_{\delta + \varepsilon}^k}\norm[op]{R_\lambda(A) - R_\lambda(A_n)}\cdot 2\pi \cdot (\delta + \varepsilon)\\
        &\hspace{-2em}\leq (q-(K-1))\cdot (\delta + \varepsilon)\cdot \max_{k\geq K}\left( \sup_{\lambda\in\partial B_{\delta + \varepsilon}^k}\norm[op]{R_\lambda(A) - R_\lambda(A_n)}\right).
    \end{align}
    Furthermore, for any $k$ and $\lambda\in\partial B_{\delta + \varepsilon}^k$, the second resolvent identity yields
    \begin{align}
        \begin{split}
            \norm[op]{R_\lambda(A) - R_\lambda(A_n)} &= \norm[op]{R_\lambda(A)(A-A_n)R_\lambda(A)}\\
            &\leq \norm[op]{A-A_n}\cdot\norm[op]{R_\lambda(A)}\norm[op]{R_\lambda(A_n)}.
        \end{split}\label{eq:resolvent_convergence_1}
    \end{align}
    We now show that the supremum over $\lambda\in\partial B_{\delta + \varepsilon}^k$ in the right hand side of~(\ref{eq:resolvent_convergence_1}) is bounded.
    Since both $A$ and $A_n$ are Hermitian, it follows that their resolvent is normal and bounded for $\lambda \in \bC\setminus\sigma(A_n)$ and $\lambda\in \bC\setminus\sigma(A)$ respectively.
    The operator norm is thus given by the spectral radius,
    \begin{equation}
        \norm[op]{R_\lambda(A)} = \max_k \abs{\lambda_k(R_\lambda(A))}
        \hspace{1em}\text{and}\hspace{1em}
        \norm[op]{R_\lambda(A_n)} = \max_k \abs{\lambda_k(R_\lambda(A_n))}.
    \end{equation}
    Note that the eigenvalues of $R_\lambda(A)$ are given by $(\lambda_k(A) - \lambda)^{-1}$. To see this, let $\lambda\in\bC\setminus\sigma(A)$ and consider
    \begin{align}
        \mathrm{det}\left(R_\lambda(A) - \mu\Id\right) &= \mathrm{det}\left((A - \lambda\Id)^{-1}\left(\Id - (A - \lambda\Id)\mu\Id\right)\right)\\
        &\propto \mathrm{det}\left(\Id - (A - \lambda\Id)\mu\Id\right) = (-\mu)^m\mathrm{det}\left(A - (\mu^{-1} + \lambda)\Id\right).
    \end{align}
    Since $\mathrm{det}(R_\lambda(A)) \neq 0$ it follows that $\mu=0$ can not be an eigenvalue. Thus, eigenvalues of $R_\lambda(A)$ satisfy
    \begin{equation}
        \frac{1}{\mu} + \lambda = \lambda_k(A) \Rightarrow \mu = \frac{1}{\lambda_k(A) - \lambda}.
    \end{equation}
    The same reasoning yields an expression for eigenvalues of $R_\lambda(A_n)$. Thus
    \begin{equation}
        \norm[op]{R_\lambda(A)} = \max_k \frac{1}{\left|\lambda_k(A) - \lambda\right|}
        \hspace{1em}\text{and}\hspace{1em}
        \norm[op]{R_\lambda(A_n)} = \max_k \frac{1}{\left|\lambda_k(A_n) - \lambda\right|}.
    \end{equation}
    Note that, by the definition of $\delta$, for $\lambda\in\partial B_{\delta + \varepsilon}^k$, the eigenvalue of $A$ which is nearest to $\lambda$ is given by $\lambda_k(A)$. 
    Since this is exactly the center of the ball $B_{\delta + \varepsilon}$, it follows that
    \begin{equation}
        \sup_{\lambda \in \partial B_{\delta + \varepsilon}^k}\norm[op]{R_\lambda(A)} = \frac{1}{\delta + \varepsilon} < \infty.
    \end{equation}
    Similarly, for $\lambda\in\partial B_{\delta + \varepsilon}^k$, the eigenvalue of $A_n$ which is nearest to $\lambda$ is given $\lambda_k(A_n)$ since $n$ was chosen large enough such that $\left|\lambda_k(A_n) - \lambda_k(A)\right| < \varepsilon$ and $\varepsilon < r_0 / 2$.
    Since $\delta < \frac{r_0}{2} - \varepsilon$, it follows that the smallest distance from $\partial B_{\delta + \varepsilon}^k$ to $\lambda_k(A_n)$ is exactly $\delta$ and thus
    \begin{equation}
        \sup_{\lambda \in \partial B_{\delta + \varepsilon}^k}\norm[op]{R_\lambda(A_n)} = \frac{1}{\delta} < \infty.
    \end{equation}
    Hence, we find that the RHS in~(\ref{eq:resolvent_convergence_1}) is bounded by $\frac{\norm[op]{A_n - A}}{(\delta + \varepsilon)\delta}$ for $\lambda\in\partial B_{\delta + \varepsilon}^k$.
    Finally, this yields
    \begin{align}
        \norm[op]{\{A_n < 0\} - \{A < 0\}} &\leq
        (\delta + \varepsilon) (q - (K-1)) \max_{k \geq K} \left(\sup_{\lambda\in\partial B_{\delta + \varepsilon}^k}\norm[op]{R_\lambda(A) - R_\lambda(A_n)}\right)\\
        &\leq (\delta + \varepsilon) (q - (K-1)) \norm[op]{A_n - A}\frac{1}{\delta + \varepsilon}\frac{1}{\delta}\\
        &= \frac{q-(K-1)}{\delta}\norm[op]{A_n - A} \xrightarrow{n\to\infty}0.
    \end{align}
    In an analogous way we can show that
    \begin{align}
        \norm[op]{\{A_n > 0\} - \{A > 0\}} \leq \frac{R}{\delta}\norm[op]{A_n - A} \xrightarrow{n\to\infty}0
    \end{align}
    where $R$ denotes the index of the smallest positive eigenvalue of $A$.
    This concludes the proof.
\end{proof}
\begin{lem}
    \label{lem:alpha_non_increasing}
    The functions $t\mapsto \alpha(P_{t,+})$ and $t\mapsto \alpha(P_{t,+} + P_{t,0})$ are non-increasing.
\end{lem}
\begin{proof}
    Let $0 \leq s < t$. We need to show that $\Tr{\sigma P_{s,+}} \geq \Tr{\sigma P_{t,+}}$ and $\Tr{\sigma (P_{s,+}+P_{s,0})} \geq \Tr{\sigma (P_{t,+}+P_{t,0})}$. Note that for any Hermitian operator $A$ on and for any Hermitian operator $0 \leq T \leq \Id$ we have
    \begin{equation}
        \Tr{A\{A \geq 0\}} = \Tr{A\{A > 0\}} \geq \Tr{A\cdot T}
    \end{equation}
    We first show $\Tr{\sigma P_{s,+}} \geq \Tr{\sigma P_{t,+}}$. Write $T_s := \rho - s\sigma$ and $T_t := \rho - t\sigma$ and note that the eigenprojections are Hermitian and satisfy
    \begin{equation}
        0 \leq \{T_s > 0\}\leq \Id,
        \ws\text{and}\ws
        \ws 0 \leq \{T_t > 0\} \leq \Id
    \end{equation}
    It follows that
    \begin{equation}
        \label{eq:monotonicity1}
        \Tr{T_t\{T_t > 0\}} \geq \Tr{T_t\{T_s > 0\}}
    \end{equation}
    and similarly
    \begin{equation}
        \label{eq:monotonicity2}
        \Tr{T_s\{T_s > 0\}} \geq \Tr{T_s\{T_t > 0\}}.
    \end{equation}
    Combining~(\ref{eq:monotonicity1}) and~(\ref{eq:monotonicity2}) yields
    \begin{align}
            t\cdot \left(\Tr{\sigma\{T_s > 0\}} - \Tr{\sigma\{T_t > 0\}}\right)&\geq
            s\cdot\left(\Tr{\sigma\{T_s > 0\}} - \Tr{\sigma\{T_t > 0\}}\right).
    \end{align}
    Since $s < t$ it follows that $\Tr{\sigma\{T_s > 0\}} \geq \Tr{\sigma\{T_t > 0\}}$ and thus $\alpha(P_{s,+}) \geq \alpha(P_{t,+})$.
    The other statement follows analogously by replacing $\{T_t > 0\}$ and $\{T_s > 0\}$ by $\{T_t \geq 0\}$ and $\{T_s \geq 0\}$. 
    This concludes the proof.
\end{proof}
\begin{lem}
    \label{lem:alpha_right_continuous}
    The function $t\mapsto \alpha(P_{t,+})$ is continuous from the right.
\end{lem}
\begin{proof}
    Let $t\geq 0$ and let $\{t_n\}_{n\in\bN}\subseteq[0,\,\infty)$ be a sequence such that $t_n\downarrow t$ (i.e. $t_n$ converges to $t$ from above). We show that $\lim_{n\to\infty}\alpha(P_{t_n,+}) = \alpha(P_{t,+})$. Define the operators
    \begin{equation}
        A_n := \rho -t_n\sigma,
        \ws
        A := \rho -t\sigma
    \end{equation}
    and note that
    \begin{equation}
        \norm[op]{A_n-A} = \abs{t_n-t}\cdot\norm[op]{\sigma} \xrightarrow{n\to\infty}0.
    \end{equation}
    Since, in addition, both $A_n$ and $A$ are Hermitian and $A - A_n = (t_n - t)\sigma \geq 0$, it follows from the second part of Lemma~\ref{lem:operator_trace_convergence} that
    \begin{equation}
        \norm[op]{\{A_n > 0\} - \{A > 0\}} \xrightarrow{n\to\infty}0
    \end{equation}
    and thus
    \begin{equation}
        \alpha(P_{t_n,+}) = \Tr{\sigma\{A_n > 0\}} \xrightarrow{n\to\infty} \Tr{\sigma\{A > 0\}} = \alpha(P_{t,+})
    \end{equation}
    since operator norm convergence implies convergence in the weak operator topology. This concludes the proof.
\end{proof}
\begin{lem}
    \label{lem:alpha_left_continuous}
    The function $t\mapsto \alpha(P_{t,+} + P_{t,0})$ is continuous from the left.
\end{lem}
\begin{proof}
    Let $\{t_n\}_{n\in\bN}\subset[0,\infty)$ be a sequence of non-negative real numbers such that $t_n \uparrow t$ (i.e. $t_n$ converges to $t$ from below).
    Let $A_n$ and $A$ be the Hermitian operators defined by
    \begin{equation}
        A_n := \rho - t_n \sigma,\ws A := \rho - t \sigma
    \end{equation}
    and note that $A - A_n = (t_n - t)\sigma \leq 0$ and $\norm[op]{A_n - A} \rightarrow 0$ as $n\to \infty$. It follows from the first part of Lemma~\ref{lem:operator_trace_convergence} that
    \begin{align}
        \norm[op]{\{A_n < 0\}-\{A < 0\}} \xrightarrow{n\to\infty}0
    \end{align}
    and thus, since operator norm convergence implies convergence in the weak operator topology, we have
    \begin{equation}
        \alpha(P_{t_{n},+} + P_{t_n,0}) = \Tr{\sigma(\Id - \{A_n < 0\})} \xrightarrow{n\to\infty} \Tr{\sigma(\Id - \{A < 0\})} = \alpha(P_{t,+} + P_{t,0}).
    \end{equation}
    This concludes the proof.
\end{proof}
\begin{lem}
    \label{lem:sandwich}
    For $\alpha_0\in[0,\,1]$ we have the chain of inequalities
    \begin{equation}
        \alpha\left(P_{\tau(\alpha_0),+}\right) \leq \alpha_0 \leq \alpha\left(P_{\tau(\alpha_0),+} + P_{\tau(\alpha_0),0}\right).
    \end{equation}
\end{lem}
\begin{proof}
    Recall that $\tau(\alpha_0) := \inf\{t\geq 0\colon\, \alpha(P_{t,+}) \leq \alpha_0\}$. Since, by Lemmas~\ref{lem:alpha_non_increasing} and~\ref{lem:alpha_right_continuous} the function $t\mapsto\alpha(P_{t,+})$ is non-decreasing and right-continuous, the left hand side of the inequality follows directly from the definition of $\tau(\alpha_0)$.
    
    We now show the right hand side of the inequality. Note that if $t:=\tau(\alpha_0)=0$, then $P_{t,-} = \{\rho < 0\} = 0$ and thus $P_{t,+} + P_{t,0} = \Id$ and $\alpha(P_{t,+} + P_{t,0}) = \Tr{\sigma}= 1 \geq \alpha_0$. Suppose that $\tau(\alpha_0)>0$ and let $\{t_n\}_{n\in\bN}\subset[0,\infty)$ be a sequence of non-negative real numbers such that $t_n \uparrow \tau(\alpha_0)$.
    Note that, since $t_n < \tau(\alpha_0)$ for all $n$ and $t\mapsto\alpha(P_{t,+})$ is non-increasing and right-continuous by Lemmas~\ref{lem:alpha_non_increasing} and~\ref{lem:alpha_right_continuous}, we have
    \begin{equation}
        \Tr{\sigma P_{t_n,+}} = \alpha(P_{t_{n},+}) > \alpha_0
    \end{equation}
    by definition of $\tau(\alpha_0)$.
    Furthermore, since the projection $P_{t_n,0}$ is positive semidefinite, it follows that
    \begin{align}
        \alpha(P_{t_{n},+} + P_{t_n,0}) = \Tr{\sigma (P_{t_n,+} + P_{t_n,0})} \geq \Tr{\sigma P_{t_n,+}} = \alpha(P_{t_{n},+}) > \alpha_0.
    \end{align}
    The Lemma now follows since by Lemma~\ref{lem:alpha_left_continuous} the function $t\mapsto \alpha(P_{t,+} + P_{t,0})$ is continuous from the left and hence
    \begin{equation}
        \alpha(P_{\tau(\alpha_0),+} + P_{\tau(\alpha_0),0}) = \lim_{n\to\infty} \alpha(P_{t_{n},+} + P_{t_n,0}) \geq \alpha_0
    \end{equation}
    which concludes the proof.
\end{proof}

\section{Construction of Helstrom Operators and Optimality}
Given $\alpha_0\in[0,\,1]$, it follows from Lemma~\ref{lem:sandwich} that firstly $\alpha_0 - \alpha\left(P_{\tau(\alpha_0),+}\right) \geq 0$, and secondly, since $\alpha_0 - \alpha\left(P_{\tau(\alpha_0),+}\right) \leq \alpha\left(P_{\tau(\alpha_0),+} +P_{\tau(\alpha_0),0}\right) - \alpha\left(P_{\tau(\alpha_0),+}\right)$ we have that
\begin{equation}
    \frac{\alpha_0 - \alpha\left(P_{\tau(\alpha_0),+}\right)}{\alpha\left(P_{\tau(\alpha_0),0}\right)} \in [0,\,1]
\end{equation}
if $\alpha\left(P_{\tau(\alpha_0),0}\right) \neq 0$. It follows that, for any $\alpha_0\in[0,\,1]$, a Helstrom operator with type-I error probability $\alpha_0$ is given by
\begin{equation}
    \label{eq:helstrom_measurement}
    M_{\tau(\alpha_0)}:=P_{\tau(\alpha_0),+} + q_0P_{\tau(\alpha_0),0},
    \ws\ws\mathrm{where}\ws\ws
    q_0 := 
        \begin{cases}
            \frac{\alpha_0 - \alpha(P_{\tau(\alpha_0),+})}{\alpha(P_{\tau(\alpha_0),0})}\,&\mathrm{if}\ \alpha\left(P_{\tau(\alpha_0),0}\right) \neq 0,\\
            0 &\mathrm{otherwise.}
        \end{cases}
\end{equation}
To see that $M_{\tau(\alpha_0)}$ indeed has type-I error probability $\alpha_0$, note that if $q_0\neq 0$ then
\begin{equation}
    \alpha(M_{\tau(\alpha_0)}) = \Tr{\sigma M_{\tau(\alpha_0)}} = \Tr{\sigma P_{\tau(\alpha_0),+}} + q_0\Tr{\sigma P_{\tau(\alpha_0),0}} = \alpha_0.
\end{equation}
If on the other hand $q_0=0$, then, by Lemma~\ref{lem:sandwich}, we have $\alpha_0 = \alpha(P_{\tau(\alpha_0),+}) = \alpha(M_{\tau(\alpha_0)})$.
The following lemma shows optimality of the Helstrom operators, i.e. equation~\eqref{eq:helstrom_optimality}:
\begin{lem}
    \label{lem:helstrom_optimality}
    Let $t\geq0$ and let $M_t := P_{t,+} + X_t$ for $0\leq X_t \leq P_{t,0}$ be a Helstrom operator. Then, for any quantum hypothesis test $0\leq M \leq \Id$ for testing the null $\sigma$ against the alternative $\rho$, the following implications hold
    \begin{enumerate}
        \item[(i)] $\alpha(M) \leq \alpha(M_t) \ \Rightarrow \ \beta(M) \geq \beta(M_t)$
        \item[(ii)] $\alpha(M) \geq 1 - \alpha(M_t) \ \Rightarrow \ 1 - \beta(M) \geq \beta(M_t)$
    \end{enumerate}
\end{lem}
\begin{proof}
    Let $\sum_i\lambda_i P_i$ be the spectral decomposition of the operator $\rho -t \sigma$ with orthogonal projections $\{P_i\}_i$ and associated eigenvalues $\{\lambda_i\}_i$.
    Recall that
    \begin{equation}
        P_{t,+} := \sum_{i\colon\lambda_i > 0} P_i\ws
        P_{t,-} := \sum_{i\colon\lambda_i < 0} P_i,\ws
        P_{t,0} := \Id - P_{t,+} - P_{t,-}.
    \end{equation}
    We notice that for any $0\leq X_t \leq P_{t,0}$ we have
    \begin{equation}
        \Tr{(\rho - t\sigma)X_{t}} = \Tr{(\rho-t\sigma)P_{t,+}X_t} + \Tr{(\rho-t\sigma)P_{t,-}X_t} \leq \Tr{(\rho-t\sigma)P_{t,+}P_{t,0}} = 0
    \end{equation}
    and
    \begin{equation}
        \Tr{(\rho - t\sigma)X_{t}} = \Tr{(\rho-t\sigma)P_{t,+}X_t} + \Tr{(\rho-t\sigma)P_{t,-}X_t} \geq \Tr{(\rho-t\sigma)P_{t,-}P_{t,0}} = 0
    \end{equation}
    and thus $\Tr{(\rho - t\sigma)P_{t,0}} = \Tr{(\rho - t\sigma)X_{t}} = 0$. For the sequel, let $\Bar{M}_t := \Id - M_t$ and $\Bar{M} := \Id - M$.

    We first show part $(i)$ of the statement. Multiplying with the identity yields
    \begin{equation}
        M - M_t = (\Bar{M}_t + M_t)M - M_t(\Bar{M} + M) = \Bar{M}_t M - M_t\Bar{M}
    \end{equation}
    and adding zero yields
    \begin{align}
        \rho(M - M_t) &= (\rho - t\sigma)(M - M_t) + t\sigma(M - M_t)\\ 
        &=(\rho - t\sigma)(\Bar{M}_t M - M_t\Bar{M}) + t\sigma(\Bar{M}_t M - M_t\Bar{M}).
    \end{align}
    We need to show that $\beta(M_t) - \beta(M) = \Tr{\rho(M - M_t)} \leq 0$. Notice that
    \begin{equation}
        \Tr{(\rho - t\sigma)\Bar{M}_t M} = - \Tr{(\rho - t\sigma)_- M} \leq 0
    \end{equation}
    and similarly
    \begin{equation}
        \Tr{(\rho - t\sigma)M_t \Bar{M}} = \Tr{(\rho - t\sigma)_+ \Bar{M}} \geq 0
    \end{equation}
    where the inequalities follow from $\Id \geq M \geq 0$. Finally, we see that
    \begin{align}
        \Tr{\rho(M - M_t)} &= \Tr{(\rho - t\sigma)(\Bar{M}_t M - M_t\Bar{M})} + t\cdot\Tr{\sigma(\Bar{M}_t M - M_t\Bar{M})}\\
        &\leq t\cdot \Tr{\sigma(\Bar{M}_t M - M_t\Bar{M})}\\
        &= t\cdot \Tr{\sigma(M - M_t)}\\
        &= t\cdot (\alpha(M) - \alpha(M_t)) \leq 0
    \end{align}
    where the last inequality follows from the assumption and $t\geq0$.

    Part $(ii)$ now follows directly from part $(i)$ by noting that $0 \leq M' := \Id - M \leq \Id$ and
    \begin{equation}
        \alpha(M) \geq 1 - \alpha(M_t) \Rightarrow \alpha(M') \leq \alpha(M_t) {\,\overset{(i)}{\Longrightarrow}\,} \beta(M_t) \leq \beta(M') = 1 - \beta(M).
    \end{equation}
    This concludes the proof.
\end{proof}

\section{Proof of Lemma~\ref{lem:pure_state_bound}}
\begin{customlem}{\ref{lem:pure_state_bound}}[restated]
    Let $\ket{\psi_\sigma},\,\ket{\psi_\rho}\in\cH$ and let $\cA$ be a quantum classifier. Suppose that for $k_A\in\cC$ and $p_{\mathrm{A}},\,p_{\mathrm{B}}\in[0,\,1]$, we have $k_A = \cA(\psi_\sigma)$ and suppose that the score function $\y$ satisfies~\eqref{eq:class_probs}.
    Then, it is guaranteed that $\cA(\psi_\rho)=\cA(\psi_\sigma)$ for any $\psi_\rho$ with
    \begin{equation}
        \abs{\braket{\psi_\sigma}{\psi_\rho}}^2 > \frac{1}{2}\left(1 + \sqrt{g(p_{\mathrm{A}},\,p_{\mathrm{B}})}\right),
    \end{equation}
    where the function $g$ is given by
    \begin{equation}
        \begin{aligned}
            g(p_{\mathrm{A}},\,p_{\mathrm{B}}) &= 1 - p_{\mathrm{B}} - p_{\mathrm{A}}(1-2p_{\mathrm{B}}) + \\
            &\hspace{4em}2\sqrt{p_{\mathrm{A}} p_{\mathrm{B}}(1-p_{\mathrm{A}})(1-p_{\mathrm{B}})}.
        \end{aligned}
    \end{equation}
    This condition is equivalent to~\eqref{eq:robustness_condition} and is hence both sufficient and necessary whenever $p_{\mathrm{A}} + p_{\mathrm{B}} = 1$.
\end{customlem}
\begin{proof}
    In order to prove the lemma, we show that the fidelity bound in eq.~\eqref{eq:fidelity_bound_pure} is equivalent to the SDP robustness condition~\eqref{eq:robustness_condition} from Theorem~\ref{thm:main} expressed in terms of type-II error probabilities.
    For that purpose, we first derive an expression for the Helstrom operators in terms of the squared overlap between $\ket{\psi_\rho}$ and $\ket{\psi_\sigma}$ and subsequently solve~\eqref{eq:robustness_condition} for $\abs{\braket{\psi_\sigma}{\psi_\rho}}^2$.
    For the sequel, let $\sigma=\ketbra{\psi_\sigma}{\psi_\sigma}$ and $\rho=\ketbra{\psi_\rho}{\psi_\rho}$.
    Recall that a Helstrom operator with type-I error probability $\alpha_0$ takes the form~\eqref{eq:helstrom_measurement} which reads
    \begin{equation}
        M_{\tau(\alpha_0)}:=P_{\tau(\alpha_0),+} + q_0P_{\tau(\alpha_0),0},
    \ws 
    q_0 := 
        \begin{cases}
            \frac{\alpha_0 - \alpha(P_{\tau(\alpha_0),+})}{\alpha(P_{\tau(\alpha_0),0})}\,&\mathrm{if}\,\alpha\left(P_{\tau(\alpha_0),0}\right) \neq 0,\\
            0 &\mathrm{otherwise.}
        \end{cases}
    \end{equation}
    where $\tau(\alpha_0):= \inf\{t\geq 0\colon\,\alpha(P_{t,+}) \leq \alpha_0\}$. We now proceed as follows. We first compute the spectral decomposition of the operator $\rho-t\sigma$ as a function of $t$. With this, we derive an expression for $\alpha(P_{t,+})$ and subsequently compute $\tau(\alpha_0)$. This yields an expression for the Helstrom operators with type-I error probabilities $1-p_{\mathrm{A}}$ and $p_{\mathrm{B}}$ which can then be used to solve inequality~\eqref{eq:robustness_condition} for the fidelity. 
    We thus start by solving the eigenvalue problem
    \begin{equation}
        \label{eq:eigenvalue_problem_pure}
        (\rho - t \sigma)\ket{\eta} = \eta \ket{\eta}.
    \end{equation}
    Since $\sigma$ and $\rho$ are both pure states, the operator
    $\rho - t\sigma$ is of rank at most $2$. It follows that there are at most two states $\ket{\eta_0}$ and $\ket{\eta_1}$ satisfying~(\ref{eq:eigenvalue_problem_pure}) with nonzero eigenvalues and, in addition, they are linear combinations of ${\ket{\psi_\sigma}}$ and $\ket{{\psi_\rho}}$
    \begin{equation}
        \ket{\eta_k} = z_{k,\sigma}{\ket{\psi_\sigma}} + z_{k,\rho}{\ket{\psi_\rho}},\hspace{2em} k=0,\,1
    \end{equation}
    with constants $z_{k,\sigma}$ and $z_{k,\rho}$ that are to be determined. Substituting this into~(\ref{eq:eigenvalue_problem_pure}) yields the system of equations
    \begin{align}
        \begin{split}
            z_{k,\rho} + \gamma z_{k,\sigma} &= \eta_k z_{k,\rho}\\
            -t\gamma^\dagger z_{k,\rho} - t z_{k,\sigma} &= \eta_k z_{k,\sigma}
        \end{split}
    \end{align}
    where $\gamma := {\braket{\psi_\rho}{\psi_\sigma}}$ is the overlap between states $\rho$ and $\sigma$.
    The two eigenvalues $\eta_k$ for which these equations possess nonzero solutions are given by
    \begin{equation}
        \begin{gathered}
            \eta_0 = \frac{1}{2}\left(1 - t\right) + R, \hspace{2em}\eta_1 = \frac{1}{2}\left(1 - t\right) - R,\\
            R = \sqrt{\frac{1}{4}\left(1 - t\right)^2 + t\left(1 - \abs{\gamma}^2\right)}
        \end{gathered}
    \end{equation}
    with $\eta_0 > 0$ and $\eta_1 \leq 0$. With the condition $\braket{\eta_k}{\eta_k} = 1$, the coefficients $z_{k,\sigma}$ and $z_{k,\rho}$ are then determined as
    \begin{equation}
        \begin{gathered}
            z_{k,\rho} = -\gamma A_k,\hspace{2em} z_{k,\sigma} = (1 -\eta_k) A_k,\\
            \abs{A_k}^{-2} = 2R\abs{\eta_k - 1 + \abs{\gamma}^2}.
        \end{gathered}
    \end{equation}
    Recall that $P_{t,+}=\sum\limits_{k\colon\eta_k>0}P_k$ and hence $P_{t,+} = \ketbra{\eta_0}{\eta_0}$. We thus obtain the expression
    \begin{align}
        \alpha(P_{t,+}) &= \Tr{\sigma P_{t,+}} 
        = \abs{\braket{\eta_0}{{\psi_\sigma}}}^2\\
        &= \abs{z_{0,\sigma} + z_{0,\rho}\gamma^*}^2\\
        &= \abs{A_0}^2\abs{1-\eta_0 - \abs{\gamma}^2}^2\\
        &= \frac{\eta_0 - 1 + \abs{\gamma}^2}{2R}
    \end{align}
    Substituting in the expressions for $\eta_0$ and $R$ yields
    \begin{align}
        \alpha(P_{t,+}) &= \dfrac{\frac{1}{2}\left(1 - t\right) + \sqrt{\frac{1}{4}\left(1 - t\right)^2 + t\left(1 - \abs{\gamma}^2\right)} - 1 + \abs{\gamma}^2}{2\sqrt{\frac{1}{4}\left(1 - t\right)^2 + t\left(1 - \abs{\gamma}^2\right)}}
        = \frac{1}{2}\left(1 - \frac{1 + t - 2\abs{\gamma}^2}{\sqrt{(1 + t)^2 - 4t\abs{\gamma}^2}}\right)
    \end{align}
    The next step is to compute $t_A=\tau(1-p_{\mathrm{A}})$ and $t_B=\tau(p_{\mathrm{B}})$.
    By Lemma~\ref{lem:alpha_non_increasing}, the function $t \mapsto g(t):=\alpha(P_{t,+})$ is non-increasing and thus attains its maximum at $t=0$
    \begin{equation}
        g(0) = \alpha(P_{0,+}) = \abs{\gamma}^2
    \end{equation}
    Furthermore, note that the only real non-negative discontinuity of $g$ is located at $t=1$ in the case where $\abs{\gamma}^2=1$. Since this corresponds to two identical states we exclude this case in the following and assume $\abs{\gamma} \in [0,\,1)$.
    Notice that, if $\alpha_0\geq \abs{\gamma}^2$, then we have that $\tau(\alpha_0)=0$. Otherwise, if $\alpha_0 < \abs{\gamma}^2$, then we have that $\alpha(P_{t,+}) = \alpha_0$ if $t\geq 0$ is the non-negative root of the polynomial
    \begin{equation}
        t \mapsto t^2 + 2t(1-2\abs{\gamma}^2) + \left(1-\frac{\abs{\gamma}^2(1-\abs{\gamma}^2)}{\alpha_0(1-\alpha_0)}\right)
    \end{equation}
    which is calculated as
    \begin{equation}
        t = 2\abs{\gamma}^2 - 1 - (2\alpha_0-1)\sqrt{\frac{\abs{\gamma}^2(1-\abs{\gamma}^2)}{\alpha_0(1-\alpha_0)}}.
    \end{equation}
    Thus, in summary, we find that $\tau(\alpha_0)$ is given by
    \begin{equation}
        \tau(\alpha_0) = 
            \begin{cases}
                2\abs{\gamma}^2 - 1 - (2\alpha_0-1)\sqrt{\frac{\abs{\gamma}^2(1-\abs{\gamma}^2)}{\alpha_0(1-\alpha_0)}} \ &\ \mathrm{if}\ \alpha_0 < \abs{\gamma}^2,\\
                0 \ &\ \mathrm{if}\ \alpha_0 \geq \abs{\gamma}^2.
            \end{cases}
    \end{equation}
    First, we notice that if $\abs{\gamma}^2 \leq \min\{p_{\mathrm{B}},\,1-p_{\mathrm{A}}\}$, then we have $t_A = \tau(1-p_{\mathrm{A}}) = 0$ and $t_B= \tau(p_{\mathrm{B}}) = 0$ and hence the robustness condition~(\ref{eq:robustness_condition}) cannot be satisfied. In the case where $\min\{p_{\mathrm{B}},\,1-p_{\mathrm{A}}\} < \abs{\gamma}^2 \leq \max\{p_{\mathrm{B}},\,1-p_{\mathrm{A}}\}$, then either $t_A = 0$ or $t_B = 0$. Without loss of generality, suppose that $t_A = 0$. Then the Helstrom operator takes the form ${M_{\tau(1-p_{\mathrm{A}})}}=\{\rho > 0\} + q_A(\Id - \{\rho > 0\} - \{\rho < 0\})$ and thus
    \begin{equation}
        \beta({M_{\tau(1-p_{\mathrm{A}})}}) = 1 - \Tr{\rho\left(\{\rho > 0\} + q_A(\Id - \{\rho > 0\} - \{\rho < 0\})\right)} = 0.
    \end{equation}
    In particular, it follows that the robustness condition {$\beta^*_{1-p_{\mathrm{A}}}(\sigma,\,\rho) + \beta^*_{p_{\mathrm{B}}}(\sigma,\,\rho) > 1$ cannot be satisfied}. 
    The same follows in the case where $t_B = 0$. 
    Finally, if $\abs{\gamma}^2 > \max\{p_{\mathrm{B}},\,1-p_{\mathrm{A}}\}$, then we have that $t_A > 0$ and $t_B > 0$.
    We notice that $\alpha(P_{t_A,+})=1-p_{\mathrm{A}}$ and $\alpha(P_{t_B,+})=p_{\mathrm{B}}$ and thus $M_{\tau(1-p_{\mathrm{A}})} = P_{t_A,+}$ and $M_{\tau(p_{\mathrm{B}})} = P_{t_B,+}$. Computing the type-II error for $M_{\tau(1-p_{\mathrm{A}})}$ yields
    \begin{align}
        \beta(M_{\tau(1-p_{\mathrm{A}})}) &= 1 - \Tr{\rho P_{t_A,+}} = 1 - \abs{\braket{\eta_0}{{\psi_\rho}}}^2 = 1 - \abs{A_0}^2\abs{\gamma}^2\abs{\eta_0}^2\\
        &=\abs{\gamma}^2(2p_{\mathrm{A}}-1)+(1-p_{\mathrm{A}})\left(1-2p_{\mathrm{A}}\sqrt{\frac{\abs{\gamma}^2(1-\abs{\gamma}^2)}{p_{\mathrm{A}}(1-p_{\mathrm{A}})}}\right)
    \end{align}
    where $\ket{\eta_0}$ corresponds to the eigenvector associated with the eigenvalue $\eta_0$ at $t=t_A$. Similarly, computing the type-II error for $M_{\tau(p_{\mathrm{B}})}$ yields
    \begin{align}
        \beta(M_{\tau(p_{\mathrm{B}})}) &= 1 - \Tr{\rho P_{t_B,+}} = 1 - \abs{\braket{\eta_0}{{\psi_\rho}}}^2 = 1 - \abs{A_0}^2\abs{\gamma}^2\abs{\eta_0}^2\\
        &=\abs{\gamma}^2(1-2p_{\mathrm{B}})+p_{\mathrm{B}}\left(1-2(1-p_{\mathrm{B}})\sqrt{\frac{\abs{\gamma}^2(1-\abs{\gamma}^2)}{p_{\mathrm{B}}(1-p_{\mathrm{B}})}}\right)
    \end{align}
    where $\ket{\eta_0}$ corresponds to the eigenvector associated with the eigenvalue $\eta_0$ at $t=t_B$. With these expressions, it follows that the robustness condition~(\ref{eq:robustness_condition}), i.e. $\beta(M_{\tau(1-p_{\mathrm{A}})}) + \beta(M_{\tau(p_{\mathrm{B}})}) > 1$, is equivalent to
    \begin{equation}
         {\abs{\braket{\psi_\sigma}{\psi_\rho}}^2} > \frac{1}{2}\left(1 + \sqrt{1 - p_{\mathrm{B}} - p_{\mathrm{A}}(1-2p_{\mathrm{B}}) + 2\sqrt{p_{\mathrm{A}} p_{\mathrm{B}}(1-p_{\mathrm{A}})(1-p_{\mathrm{B}})}}\right)
    \end{equation}
    {what concludes the proof.}
\end{proof}

\section{Proof of Corollary~\ref{cor:trace_bound_pure_mixed}}
\begin{customcor}{\ref{cor:trace_bound_pure_mixed}}[restated]
    Let $\sigma,\,\rho\in\cS(\cH)$ and suppose that $\sigma = \ketbra{\psi_\sigma}{\psi_\sigma}$ is pure. Let $\cA$ be a quantum classifier and suppose that for $k_A\in\cC$ and $p_{\mathrm{A}},\,p_{\mathrm{B}}\in[0,\,1]$, we have $k_A = \cA(\sigma)$ and suppose that the score function $\y$ satisfies~\eqref{eq:class_probs}.
    Then, it is guaranteed that $\cA(\rho)=\cA(\sigma)$ for any $\rho$ with
    \begin{align}
        T(\rho,\,\sigma) < \delta(p_{\mathrm{A}},\,p_{\mathrm{B}})\left(1-\sqrt{1-\delta(p_{\mathrm{A}},\,p_{\mathrm{B}})^2}\right)
    \end{align}
    where 
    $
    \delta(p_{\mathrm{A}},\,p_{\mathrm{B}}) =[\frac{1}{2}\left(1-g(p_{\mathrm{A}}, p_{\mathrm{B}})\right)]^\frac{1}{2}.
    $
\end{customcor}
\begin{proof}
    We denote the convex hull enclosed by the set of robust pure states as $\mathcal{C}:=\text{Conv}(\{\ketbra{\psi}{\psi}: \|\ketbra{\psi}{\psi}-\sigma\|_1<\delta(P_A,P_B)\})$. 
    Observe that any convex mixture $\rho=\sum_ip_i\ketbra{\psi_i}{\psi_i}$ with $\sum_ip_i=1$ of any sets of robust pure states $\{\ketbra{\psi_i}{\psi_i}\}\in \mathcal{C}$ must also be robust. Thus it suffices to prove condition \eqref{robustcondition:puremixed} implies $\rho\in\mathcal{C}$. Note that the boundary consisting of non-extreme points (which correspond to mixed-states) of $\mathcal{C}$ interfaces with the set $\mathcal{C^*}=\text{Conv}(\{\ketbra{\psi^*}{\psi^*}:\|\ketbra{\psi^*}{\psi^*}-\sigma\|_1=\delta(P_A,P_B)\})$. Thus, it suffices to compute the shortest distance $r$ from $\sigma$ to  $\mathcal{C^*}$, such that $r=\min_{\rho^*}\|\rho^*-\sigma\|_1$ where $\rho^*\in \mathcal{C}^*$, then $\|\rho-\sigma\|_1<r$ guarantees robustness. Further note that for every $\rho^*$, $\exists D_\sigma \rho^* D_\sigma^\dagger \in \mathcal{C}^*$, where $D_\sigma=2\sigma-\Id$, such that $\|\rho^*-\sigma\|_1=\|D_\sigma \rho^* D_\sigma^\dagger-\sigma\|_1$, and $\|p_1\rho^*+p_2D_\sigma \rho^* D_\sigma^\dagger-\sigma\|_1<\|\rho^*-\sigma\|_1$ for $p_1+p_2=1$, and $p_1\neq 0$, $p_2\neq 0$. Therefore to minimise the distance to $\sigma$, it suffices to require $\rho^*=D_\sigma \rho^* D_\sigma^\dagger$, a valid of which is $\rho^*=\frac{1}{2}(\ketbra{\psi^*}{\psi^*}+D_\sigma\ketbra{\psi^*}{\psi^*}D_\sigma^\dagger)$. As such we have
    \begin{align}
    r&=\|\sigma-\frac{1}{2}(\ketbra{\psi^*}{\psi^*}+D_\sigma\ketbra{\psi^*}{\psi^*}D_\sigma^\dagger)\|_1    \nonumber\\
    &=\|\ketbra{\psi^*}{\psi^*}-\sigma+2|\braket{{\psi_\sigma}}{\psi^*}|^2\sigma - \braket{{\psi_\sigma}}{\psi^*}\ketbra{{\psi_\sigma}}{\psi} - \braket{\psi}{{\psi_\sigma}}\ketbra{\psi}{{\psi_\sigma}}\|_1
    \end{align}
    Note that we have $\|\ketbra{\psi^*}{\psi^*}-\sigma\|_1=\delta(P_A, P_B)$ by definition and that
    \begin{align}
    \|2|\braket{{\psi_\sigma}}{\psi^*}|^2\sigma - \braket{{\psi_\sigma}}{\psi^*}\ketbra{{\psi_\sigma}}{\psi} - \braket{\psi}{{\psi_\sigma}}\ketbra{\psi}{{\psi_\sigma}}|_1
    &=|\braket{{\psi_\sigma}}{\psi^*}|\Tr{\ketbra{{\psi_\sigma}}{{\psi_\sigma}}+\ketbra{\psi^*}{\psi^*}-\braket{{\psi_\sigma}}{\psi^*}\ketbra{{\psi_\sigma}}{\psi^*}-\braket{\psi^*}{{\psi_\sigma}}\ketbra{\psi^*}{{\psi_\sigma}}}\nonumber\\
    &=|\braket{{\psi_\sigma}}{\psi^*}|\Tr{(\sigma-\ketbra{\psi^*}{\psi^*})(\sigma-\ketbra{\psi^*}{\psi^*})^\dagger}\nonumber\\
    &=|\braket{{\psi_\sigma}}{\psi^*}|\|\ketbra{\psi^*}{\psi^*}-\sigma\|_1\nonumber\\
    &=\delta(P_A, P_B)\sqrt{1-\frac{\delta(P_A, P_B)^2}{4}}.
    \end{align}
    Applying the reversed triangle inequality, we finally arrive at
    \begin{align}
    r &\ge \left|\|\sigma-\ketbra{\psi^*}{\psi^*}\|_1-\|2|\braket{{\psi_\sigma}}{\psi^*}|^2\sigma - \braket{{\psi_\sigma}}{\psi^*}\ketbra{{\psi_\sigma}}{\psi} - \braket{\psi}{{\psi_\sigma}}\ketbra{\psi}{{\psi_\sigma}}|_1\right| \nonumber\\
    &= \delta(P_A, P_B)\left(1-\sqrt{1-\frac{\delta(P_A, P_B)^2}{4}}\right).    
    \end{align}
\end{proof}

\section{Robustness with Depolarised Input States}
\begin{cor}[Depolarised single-qubit pure states]
    \label{cor:single_qubit_depolarized}
    Let $\ket{\psi_\sigma},\,\ket{\psi_\rho}\in\bC^2$ be single-qubit pure sates and let $\cE^{\mathrm{dep}}_p$ be a depolarising channel with noise parameter $p\in(0,\,1)$. Then, if $p_{\mathrm{A}} > 1/2$ and $p_{\mathrm{B}}=1-p_{\mathrm{A}}$, the robustness condition~(\ref{eq:robustness_condition}) for $\cE^{\mathrm{dep}}_p(\sigma)$ and $\cE^{\mathrm{dep}}_p(\rho)$ is equivalent to
    \begin{equation}
        \label{eq:depolarized_single_qubit_trace_condition}
        \frac{1}{2}\norm[1]{\ketbra{\psi_\sigma}{\psi_\sigma} - \ketbra{\psi_\rho}{\psi_\rho}} <
            \begin{cases}
                \sqrt{\frac{1}{2}-\sqrt{\frac{2p_{\mathrm{A}}(1-p_{\mathrm{A}}) - p(1-\frac{1}{2}p)}{2(1-p)^2}}} &\mathrm{if} \ \frac{1}{2} < p_{\mathrm{A}} \leq \frac{4-6p+3p^2}{4-4p+2p^2},\\
                \sqrt{\frac{p\cdot(2-p)\cdot(1-2p_{\mathrm{A}})^2}{8(1-p)^2\cdot(1-p_{\mathrm{A}})}} &\mathrm{if} \ \frac{4-6p+3p^2}{4-4p+2p^2} < p_{\mathrm{A}} \leq \frac{4-3p}{4-2p},\\
                1 &\mathrm{if} p_{\mathrm{A}} > \frac{4-3p}{4-2p}.\ 
            \end{cases}
    \end{equation}
\end{cor}
\begin{proof}
    In order to prove the corollary, we proceed in a manner analogous to the proof of Lemma~\ref{lem:pure_state_bound}. 
    Specifically, we show that the condition on the trace distance in eq.~\eqref{eq:depolarized_single_qubit_trace_condition} is equivalent to the SDP robustness condition~\eqref{eq:robustness_condition} from Theorem~\ref{thm:main} expressed in terms of type-II error probabilities.
    Let $\sigma=\ketbra{\psi_\sigma}{\psi_\sigma},\,\rho=\ketbra{\psi_\rho}{\psi_\rho}$ and
    recall that the type-I and type-II error probabilities are given by
    \begin{equation}
        \alpha(M;\,\cE^{\mathrm{dep}}_p(\sigma)) = \Tr{M\cE^{\mathrm{dep}}_p(\sigma)},
        \ws
        \beta(M;\cE^{\mathrm{dep}}_p(\rho)) = \Tr{(\Id - M)\cE^{\mathrm{dep}}_p(\rho)}
    \end{equation}
    with $0 \leq M \leq \Id_d$.
    Let $\sigma' := \cE^{\mathrm{dep}}_p(\sigma)$ and $\rho':=\cE^{\mathrm{dep}}_p(\rho)$ and
    recall that a Helstrom operator for testing the null $\sigma'$ against the alternative $\rho'$ with type-I error probability $\alpha_0$ takes the form~\eqref{eq:helstrom_measurement}
    \begin{equation}
        M_{\tau(\alpha_0)}:=P_{\tau(\alpha_0),+} + q_0P_{\tau(\alpha_0),0},
        \ws 
        q_0 := 
        \begin{cases}
            \frac{\alpha_0 - \alpha(P_{\tau(\alpha_0),+})}{\alpha(P_{\tau(\alpha_0),0})}\,&\mathrm{if}\,\alpha\left(P_{\tau(\alpha_0),0}\right) \neq 0,\\
            0 &\mathrm{otherwise.}
        \end{cases}
    \end{equation}
    where $\tau(\alpha_0):= \inf\{t\geq 0\colon\,\alpha(P_{t,+}) \leq \alpha_0\}$.
    Let $M_{\mathrm{A}}^\star:=M_{\tau(1-p_{\mathrm{A}})}$ and $M_{\mathrm{B}}^\star:=M_{\tau(p_{\mathrm{B}})}$ and note that by assumption $p_{\mathrm{B}} = 1-p_{\mathrm{A}}$ and hence $M_{\mathrm{A}}^\star=M_{\mathrm{B}}^\star$. The SDP robustness condition then simplifies to $\beta^*_{1-p_{\mathrm{A}}}(\sigma',\,\rho') > 1/2$.
    We now proceed as follows. We first compute the spectral decomposition of the operator $\rho'-t\sigma'$ as a function of $t$ and relate it to the fidelity between $\sigma$ and $\rho$. 
    With this, we derive an expression for $\alpha(P_{t,+})$ and subsequently compute $\tau(\alpha_0)$. 
    This yields an expression for the Helstrom operator with type-I error probability $1-p_{\mathrm{A}}$ which can then be used to solve inequality~(\ref{eq:robustness_condition}) for the fidelity.
    We thus start by solving the eigenvalue problem
    \begin{equation}
        (\rho' - t \sigma')\ket{\mu} = \mu \ket{\mu}
    \end{equation}
    which can be rewritten as
    \begin{equation}
        \left((1-p)\cdot(\rho - t\sigma) + \frac{p(1-t)}{2}\Id_2\right)\ket{\eta}.
    \end{equation}
    We notice that the operators $\rho' - t\sigma'$ and $\rho - t\sigma$ share the same set of eigenvectors. Furthermore, if $\eta$ is an eigenvalue of $\rho-t\sigma$ with eigenvector $\ket{\eta}$, then the corresponding eigenvalue $\mu$ of $\rho' - t\sigma'$ is given by
    \begin{equation}
        \mu = (1-p)\eta + \frac{p\cdot(1-t)}{2}.
    \end{equation}
    From the proof of {Lemma~\ref{lem:pure_state_bound}}, we know that the eigenvalues of $\rho-t\cdot\sigma$ are given by
    \begin{equation}
        \begin{gathered}
            \eta_0 = \frac{1}{2}(1-t)+R > 0,\ws \eta_1=\frac{1}{2}(1-t)-R \leq 0\\
            R=\sqrt{\frac{1}{4}(1-t)^2 + t(1-\abs{\gamma}^2)},\ws \gamma=\braket{{\psi_\rho}}{{\psi_\sigma}}
        \end{gathered}
    \end{equation}
    with eigenvectors
    \begin{equation}
        \begin{gathered}
            \ket{\eta_0} = -\gamma A_0 \ket{{\psi_\rho}} + (1 - \eta_0) A_0 \ket{{\psi_\sigma}},\ws \ket{\eta_2} = -\gamma A_2 \ket{{\psi_\rho}} + (1 - \eta_2) A_2 \ket{{\psi_\sigma}}\\
            \abs{A_k}^{-2} = 2R\abs{\eta_k - 1 + \abs{\gamma}^2}.
        \end{gathered}
    \end{equation}
    With this, we can compute the eigenvalues $\mu_k$ and eigenprojections $P_k$ of $\rho' - t\sigma'$ as
    \begin{equation}
        \begin{gathered}
            \mu_0 = (1 - p)\eta_0 + p\cdot \frac{1-t}{2},
            \ws 
            \mu_1 = (1 - p)\eta_1 + p\cdot \frac{1-t}{2},\\
            P_0 = \ketbra{\eta_0}{\eta_0}, 
            \ws 
            P_1 = \ketbra{\eta_1}{\eta_1}.
        \end{gathered}
    \end{equation}
    Since $\eta_0 > 0 \geq \eta_1$ for any $t\geq 0$, we have $\mu_0 \geq \mu_1$ and furthermore, the eigenvalues are monotonically decreasing functions of $t$ for $\abs{\gamma}^2 < 1$. To see this, consider
    \begin{align}
        \frac{dR}{dt} = \frac{1+t-2\abs{\gamma}^2}{2\sqrt{(1+t)^2-4t\abs{\gamma}^2}}
    \end{align}
    and thus for $\forall\,t\geq0$ and $\abs{\gamma}^2 < 1$
    \begin{align}
        \frac{d\mu_0}{dt} = \frac{dR}{dt} - \frac{1}{2} < 0
        \ws
        \text{and}
        \ws
        \frac{d\mu_1}{dt} = -\frac{dR}{dt} - \frac{1}{2} < 0.
    \end{align}
    Hence, since both eigenvalues are strictly positive at $t=0$, there exists exactly one $\xi_k$ such that $\mu_k$ vanishes at $\xi_k$, $k=0,1$. 
    Algebra shows that these zeroes are given by
    \begin{equation}
        \begin{gathered}
            \xi_0 = 1 + \frac{2(1-\abs{\gamma}^2)(1-p)^2}{p(2-p)}\left(1 + \sqrt{1+\frac{p(2-p)}{(1-\abs{\gamma}^2)(1-p)^2}}\right) > 1,\\
            \xi_1 = 1 + \frac{2(1-\abs{\gamma}^2)(1-p)^2}{p(2-p)}\left(1 - \sqrt{1+\frac{p(2-p)}{(1-\abs{\gamma}^2)(1-p)^2}}\right) < 1.
        \end{gathered}
    \end{equation}
    We define the functions
    \begin{align}
        g_0(t) &:= \bra{\eta_0}\sigma'\ket{\eta_0} = \frac{1}{2}\left(1 + \frac{(1-p)(2\abs{\gamma}^2-1-t)}{\sqrt{(1+t)^2 - 4t\abs{\gamma}^2}}\right),\\
        g_1(t)&:= \bra{\eta_1}\sigma'\ket{\eta_1} = \frac{1}{2}\left(1 - \frac{(1-p)(2\abs{\gamma}^2-1-t)}{\sqrt{(1+t)^2 - 4t\abs{\gamma}^2}}\right),\\
        f_0(t)&:= \bra{\eta_0}\rho'\ket{\eta_0} = \frac{1}{2}\left(1 + \frac{(1-p)(1+t\cdot(1-2\abs{\gamma}^2)}{\sqrt{(1+t)^2 - 4t\abs{\gamma}^2}}\right),\\
        f_1(t)&:= \bra{\eta_0}\rho'\ket{\eta_0} = \frac{1}{2}\left(1 - \frac{(1-p)(1+t\cdot(1-2\abs{\gamma}^2)}{\sqrt{(1+t)^2 - 4t\abs{\gamma}^2}}\right).
    \end{align}
    With this, we now compute $t\mapsto\alpha(P_{t,+})$ as 
    \begin{equation}
        \begin{gathered}
        \alpha(P_{t,+}) = \Tr{\sigma' P_{t,+}} = 
            \begin{cases}
                1 &0 \leq t < \xi_1\\
                g_0(t) &\xi_1 \leq t < \xi_0\\
                0 & \xi_0 \leq t
            \end{cases}
        \end{gathered}
    \end{equation}
    For $\alpha_0\in[0,\,1]$, we compute $\tau(\alpha_0) := \inf\{t\geq 0\colon\,\alpha(P_{t,+})\leq \alpha_0\}$ as
    \begin{equation}
        \tau(\alpha_0) =
            \begin{cases}
                \xi_0 & 0 \leq \alpha_0 \leq g_0(\xi_0)\\
                g_0^{-1}(\alpha_0)& g_0(\xi_0) < \alpha_0 < g_0(\xi_1)\\
                \xi_1& g_0(\xi_1) \leq \alpha_0 < 1\\
                0 &\alpha_0=1
            \end{cases}
    \end{equation}
    where
    \begin{equation}
        g_0^{-1}(\alpha_0) = 2\abs{\gamma}^2 -1 + 2(1-2\alpha_0)\sqrt{\frac{\abs{\gamma}^2(1-\abs{\gamma}^2)}{p(2-p) - 4\alpha_0(1-\alpha_0)}}.
    \end{equation}
    To solve condition~(\ref{eq:robustness_condition}) we now have to distinguish different cases, depending on which interval $1-p_{\mathrm{A}}$ falls into. Firstly, if $1-p_{\mathrm{A}} = 1$, then $\tau(1-p_{\mathrm{A}}) = 0$ and thus $\beta(M_{\mathrm{A}}^\star)=0$ in which case the condition can not be satisfied. If $1-p_{\mathrm{A}} \in [g_0(\xi_1),\,1)$, then we have $\tau(1-p_{\mathrm{A}}) = \xi_1$. In this case, it holds that $\mu_0 > 0$ and $\mu_1=0$ and the Helstrom operator is given by
    \begin{equation}
        M_{\mathrm{A}}^\star = \ketbra{\eta_0}{\eta_0} + \frac{1-p_{\mathrm{A}} - g_0(\xi_1)}{g_1(\xi_1)}\ketbra{\eta_1}{\eta_1}
    \end{equation}
    and the robustness condition reads
    \begin{equation}
        \beta(M_{\mathrm{A}}^\star) = 1 - f_0(\xi_1) - \frac{1-p_{\mathrm{A}} - g_0(\xi_1)}{g_1(\xi_1)}f_1(\xi_1) > \frac{1}{2}
    \end{equation}
    which cannot to be satisfied simultaneously with $1-p_{\mathrm{A}} \in [g_0(\xi_1),\,1)$. If, on the other hand $1-p_{\mathrm{A}}\in(g_0(\xi_0),\,g_0(\xi_1))$, then $\tau(1-p_{\mathrm{A}}) = g_0^{-1}(1-p_{\mathrm{A}})$ and $\mu_0 > 0 > \mu_1$. The Helstrom operator then is given by $M_{\mathrm{A}}^\star = \ketbra{\eta_0}{\eta_0}$ which leads to the robustness condition
    \begin{equation}
        1 - f_0(g_0^{-1}(1-p_{\mathrm{A}})) > \frac{1}{2}
    \end{equation}
    which, together with $1-p_{\mathrm{A}}\in(g_0(\xi_0),\,g_0(\xi_1))$, is equivalent to
    \begin{equation}
        \label{eq:depolarising_condition1_proof}
        \abs{\gamma}^2 > \frac{1}{2}\left(1 + \sqrt{\frac{4p_{\mathrm{A}}(1-p_{\mathrm{A}})-p(2-p)}{(1-p)^2}}\right),\ws \frac{1}{2} \leq p_{\mathrm{A}} \leq \frac{4-6p+3p^2}{4-4p+2p^2}.
    \end{equation}
    In the last case where $1-p_{\mathrm{A}}\leq g_0(\xi_0)$, we have $\tau(1-p_{\mathrm{A}}) = \xi_0$ and thus $\mu_0 = 0 > \mu_1$. The Helstrom operator is then given by $\frac{1-p_{\mathrm{A}}}{g_0(\xi_0)}\ketbra{\eta_0}{\eta_0}$, leading to the robustness condition
    \begin{equation}
        1 - \frac{1-p_{\mathrm{A}}}{g_0(\xi_0)}f_0(\xi_0) > \frac{1}{2}.
    \end{equation}
    Together with $1-p_{\mathrm{A}}\leq g_0(\xi_0)$ this is equivalent to
    \begin{equation}
        \label{eq:depolarising_condition2_proof}
        \abs{\gamma}^2 > 
            \begin{cases}
                \frac{4p_{\mathrm{A}}(1-p_{\mathrm{A}}) - p(2-p)}{(1-p)^2(4(1-p_{\mathrm{A}}) - p(2-p))}, \ &\mathrm{if}\ \frac{1 + (1-p)^2}{2} < p_{\mathrm{A}} \leq \frac{4-6p+3p^2}{4-4p+2p^2}\\
                \frac{(4 - 3p - 2p_{\mathrm{A}}(2-p))(2-p(3-2p_{\mathrm{A}}))}{8(1-p)^2(1-p_{\mathrm{A}})}, \ &\mathrm{if}\ \frac{4-6p+3p^2}{4-4p+2p^2} < p_{\mathrm{A}} \leq \frac{4-3p}{4-2p},\\
                0,\,&\mathrm{if}\, p_{\mathrm{A}} > \frac{4-3p}{4-2p}.
            \end{cases}
    \end{equation}
    Finally, combining together conditions~(\ref{eq:depolarising_condition1_proof}) and~(\ref{eq:depolarising_condition2_proof}) leads to
    \begin{equation}
        \abs{\gamma}^2 > 
            \begin{cases}
                \frac{1}{2}\left(1 + \sqrt{\frac{4p_{\mathrm{A}}(1-p_{\mathrm{A}})-p(2-p)}{(1-p)^2}}\right), \ &\mathrm{if}\ \frac{1}{2} < p_{\mathrm{A}} \leq \frac{4-6p+3p^2}{4-4p+2p^2}\\
                \frac{(4 - 3p - 2p_{\mathrm{A}}(2-p))(2-p(3-2p_{\mathrm{A}}))}{8(1-p)^2(1-p_{\mathrm{A}})}, \ &\mathrm{if}\ \frac{4-6p+3p^2}{4-4p+2p^2} < p_{\mathrm{A}} \leq \frac{4-3p}{4-2p},\\
                0,\,&\mathrm{if}\, p_{\mathrm{A}} > \frac{4-3p}{4-2p}.
            \end{cases}
    \end{equation}
    Since by assumption $\rho$ and $\sigma$ are pure states the proof is completed by noting that we have
    \begin{equation}
        T(\rho,\,\sigma) = \sqrt{1-F(\rho,\,\sigma)}
    \end{equation}
    by the Fuchs-van de Graaf inequality.
\end{proof}

\section{Algorithms}
Here, we provide pseudocode for the algorithm presented in Section~\ref{subsec:robustness_certification}.
\begin{algorithm}[H]
    \caption{Robustness Certification$(\sigma,\,N,\alpha,\,\cA)$}
    \label{alg:robustness_certification}
    \begin{algorithmic}[1]
        \REQUIRE Quantum state $\sigma\in\cS(\cH)$, number of measurement shots $N$, error tolerance $\alpha$, a quantum classifier $\cA=(\cE,\,\{\Pi_k\}_{k\in\cC})$.
        \ENSURE Predicted class $k_A$, prediction score $p_{\mathrm{A}}$ and robust radius $r_F$ according to Eq.~\eqref{eq:robustness_bound_fidelity} in terms of fidelity.
        \STATE Set counter $\mathbf{n}_k \leftarrow 0$ for every $k\in\cC$.
        \FOR{$k=1,\,\ldots\,N$}
            \STATE Apply quantum circuit $\cE$ to initial state $\sigma$.
            \STATE Perform $\abs{\cC}$-outcome measurement $\{\Pi_k\}_{k\in\cC}$ on the evolved state $\cE(\sigma)$.
            \STATE Record measurement outcome $k$ by setting $\mathbf{n}_k \leftarrow \mathbf{n}_k + 1$.
        \ENDFOR
        \STATE Calculate empirical probability distribution $\hat{\y}_k^{(N)}  \leftarrow \mathbf{n}_kN^{-1}$.
        \STATE Extract the most likely class $k_A  \leftarrow \arg\max_k \hat{\y}_k^{(N)}$.
        \STATE Set $p_{\mathrm{A}}  \leftarrow \hat{\y}_{k_A}^{(N)}(\sigma) - \sqrt{\frac{-\log(\alpha)}{2N}}$.
        \IF{$p_{\mathrm{A}} > 1/2$}
        \STATE Calculate robust radius $r_F  \leftarrow \frac{1}{2} + \sqrt{p_{\mathrm{A}}(1-p_{\mathrm{A}})}$.
        \RETURN prediction $k_A$, class score $p_{\mathrm{A}}$, robust radius $r_F$.
        \ELSE
        \RETURN ABSTAIN
        \ENDIF
    \end{algorithmic}
\end{algorithm}

\end{document}